\documentclass[a4paper]{article}
\usepackage{amssymb}
\usepackage{amsmath}
\usepackage{bm}

\newtheorem{lemma}{Lemma}
\newtheorem{theorem}{Theorem}
\newtheorem{corollary}{Corollary}
\newtheorem{observation}{Observation}

\newtheorem{example}{Example}
\newenvironment{proof}{\noindent\textit{Proof}.}{\hfill$\Box$}

\newcommand{\hide}[1]{}

\usepackage{color}

\usepackage[mathlines]{lineno}
\usepackage[normalem]{ulem}

\title{Minimum algorithm sizes for the gathering 
and related problems of autonomous mobile robots}

\author{Yuichi Asahiro and Masafumi Yamashita}

\date{August 28, 2023}

\usepackage[dvipdfmx]{graphicx}	%
\begin{document}

\maketitle

\begin{abstract}
This paper investigates a swarm of autonomous mobile robots in the Euclidean plane,
under the semi-synchronous ($\cal SSYNC$) scheduler.
Each robot has a target function to determine a destination point 
from the robots' positions.
All robots in the swarm take the same target function conventionally.
We allow the robots to take different target functions,
and investigate the effects of the number of distinct target functions 
on the problem-solving ability,
regarding target function as a resource to solve a problem like time.
Specifically, we are interested in how many distinct target functions
are necessary and sufficient to solve a problem $\Pi$.
The number of distinct target functions necessary 
and sufficient to solve $\Pi$ is called 
the {\em minimum algorithm size} (MAS) for $\Pi$.
The MAS is defined to be $\infty$, 
if $\Pi$ is unsolvable even for the robots with unique target functions.

We show that the problems form an infinite hierarchy with respect to their MASs;
for each integer $c > 0$ and $\infty$,
the set of problems whose MAS is $c$ is not empty,
which implies that target function is a resource irreplaceable, e.g., with time.
We propose MAS as a natural measure to measure the complexity of a problem.

We establish the MASs for solving the gathering and related problems 
from any initial configuration, i.e., in a self-stabilizing manner.
For example, the MAS for the gathering problem is 2.
It is 3, for the problem of gathering {\bf all non-faulty} robots at a single point, 
regardless of the number $(< n)$ of crash failures.
It is however $\infty$,
for the problem of gathering {\bf all robots} at a single point, 
in the presence of at most one crash failure.
\end{abstract}

\section{Introduction}
\label{Sintroduction}

Swarms of anonymous oblivious mobile robots have been attracting 
many researchers for three decades, e.g., 
\cite{AP04,AOSY99,BDT13,CFPS12,CP05,CDFH11,DFSY15,
DPT19,DP09,Flocchini19,FPS12,FPSW99,ISKI12,Katreniak11,Prencipe19,
SY99,Viglietta19,YS10,YUKY17}.
An anonymous oblivious mobile robot, 
which is represented by a point that moves in the Euclidean space, 
looks identical and indistinguishable, 
lacks identifiers (i.e., anonymous) and communication devices,
and operates in Look-Compute-Move cycles:
When a robot starts a cycle,
it identifies the multiset of the robots' positions,
computes the destination point using a function
called {\em target function} as in \cite{AP04,CP05,CDFH11},
based only on the multiset identified,
and then moves towards the destination point.

All the papers listed above assume that all anonymous robots 
in a swarm take the same target function.
It makes sense:
Roughly speaking, 
robots taking different target functions can behave,
as if they had different identifiers,
regarding the target functions as their identifiers (which may not be unique).
On the other hand, 
robots with different identifiers can behave,
as if they had different target functions,
even when they have the same target function.
They investigate a variety of problems from simple ones like
the convergence and the gathering problems (e.g., \cite{AP04,SY99})
to hard ones like the formation problem of a sequence of patterns 
and the gathering problem in the presence of Byzantine failures 
(e.g., \cite{DFSY15,Flocchini19}),
and show that swarms of anonymous oblivious robots are powerful enough 
to solve certain sufficiently hard problems.
At the same time, however, 
we have also recognized limitation of their problem-solving ability,
e.g., the gathering problem for two robots is not solvable \cite{YS10}.

A natural and promising approach to increase the problem-solving ability 
of a swarm is to allow robots to take different target functions,
or equivalently, to give robots different identifiers,
since almost all artificial distributed systems enjoy having unique identifiers, 
e.g., serial numbers or product numbers, to solve hard problems.
We take this approach,
and investigate the effects of the number of distinct target functions
on the problem-solving ability.
That is, we regard target function as a resource like time,
and investigate the number of distinct target functions 
necessary and sufficient to solve a problem.

Let $\cal R$ and $\Phi$ be a swarm of $n$ robots $r_1, r_2, \ldots , r_n$,
and a set of target functions such that $|\Phi| \leq n$, respectively.
A (target function) assignment ${\cal A}: {\cal R} \rightarrow \Phi$ is 
a {\bf surjection} from $\cal R$ to $\Phi$,
i.e., every target function is assigned to at least one robot.
We call $\Phi$ an {\em algorithm}\footnote{
Here, we abuse a term ``algorithm'' in two ways:
First, all robots are often assumed to take the same target function,
so that an algorithm is usually a target function that solves a problem.
Second and more importantly,
an algorithm must have a finite description.
However, a target function (and hence a set of target functions) may not,
as defined in Section 2.
To compensate this abuse,
when we will show the existence of an algorithm,
we insist on giving a finite procedure to compute it.}
of $\cal R$ to solve a problem $\Pi$,
if $\cal R$ solves $\Pi$, 
no matter which assignment $\cal A$ that $\cal R$ takes.
(Thus, we cannot assume a particular assignment 
when designing target functions.)
The size $|\Phi|$ of an algorithm $\Phi$ is
the number of target functions composing $\Phi$.
The {\em minimum algorithm size} (MAS) for $\Pi$ is the algorithm size
necessary and sufficient to solve $\Pi$.
The MAS for $\Pi$ is defined to be $\infty$,
if $\Pi$ is unsolvable even for the robots with unique target functions.

Under the semi-synchronous ($\cal SSYNC$) scheduler defined in Subsection~\ref{SSmodel},
we establish the MASs of {\bf self-stabilizing} algorithms for solving 
the gathering and related problems from {\bf any} initial configuration.
Note that the papers given in the first paragraph propose a variety of algorithms
for anonymous oblivious robots for a variety of problems,
but most of them are not self-stabilizing algorithms;
they solve problems only from initial configurations satisfying some conditions.
In what follows, an algorithm means a self-stabilizing algorithm,
unless otherwise stated.

\medskip
\noindent
{\bf Motivations.}
As pointed out, the MAS for a problem (of an anonymous swarm)
is equal to the number of identifiers (of a homonymous swarm)
necessary and sufficient to solve the problem.
We claim that target function (and hence identifier) are irreplaceable resources,
and propose that MAS is a natural complexity measure worth investigating.
To this end, 
we shall prove that there is a problem with MAS being $c$ for each integer $c > 0$ and $\infty$,
i.e., the problems form an infinite hierarchy with respect to their MASs,
which implies that target function (and hence identifier) are resources
not substitutable, e.g., with time.
We will mainly establish the MASs for the gathering and related problems,
hoping that we can reach deep understanding of their similarity and/or difference,
e.g., the reason why any algorithm for a problem needs more target functions than others.
(Indeed, the proofs are instantiations of the understanding.)

Through the research of anonymous mobile robot systems,
we have continued to pursuit a purely distributed algorithm (of size 1)
which does not rely on a central controller, as in {\bf Related works} below.
However, if it is impossible, 
it makes sense for us to pursuit an algorithm whose algorithm size is as small as possible,
since the identifier management of a homonymous swarm becomes a highly centralized task,
as the number of identifiers (i.e., the algorithm size) increases.

\medskip
\noindent
{\bf Related works -- Homonymous distributed systems.~}
A distributed system is said to be {\em homonymous},
if some processing elements (e.g., processors, processes, agents, or robots)
may have the same identifier.
Two extreme cases are anonymous systems 
and systems whose processing elements have unique identifiers.

Maintaining the system elements of a distributed system to have unique identifiers is 
a promising strategy to efficiently solve distributed problems.
All artificial distributed systems enjoy unique identifiers 
called product number, serial number, and so on, besides few exceptions.
As a matter of course,
distributed system models such as message-passing and shared memory models
assume the processing elements with unique identifiers \cite{Lynch96}.

Angluin \cite{Angluin80} started investigation on anonymous computer network in 1980,
and a few researchers (e.g., \cite{ASW88,MA89,YK88}) followed her,
to pursuit a purely distributed algorithm which does not rely on a central controller, 
in the spirit of the minimalist.
Their main research topic was symmetry breaking;
they searched for a condition symmetry breaking becomes possible
in terms of the network topology.
(Later, \cite{BV01} and \cite{YK96} characterized the solvable cases.)
A rough conclusion established is that, 
symmetry breaking is impossible in general, 
but the probability that it is possible approaches to 1,
as the number of processors increases, 
provided that the network topology is random.

Another popular problem in the early years was the function computation problem
\cite{AS91,ASW88,BV97,KKB90,YK96b}.
They characterized functions computable on anonymous networks,
and analyzed their computational complexities.
Since then, many articles have appeared on anonymous computing.

Yamashita and Kameda \cite{YK89} investigated the leader election problem
on homonymous computer networks in 1989,
and showed that the identifiers, even if they are not unique, are frequently
crucial information to solve the problem by characterizing 
when the problem becomes solvable.
The leader election problem on homonymous computer networks 
(under different problem settings) has also been investigated, 
e.g., in \cite{ADDDL20,DP04,DFT14,YK99}.

Other research topics on homonymous computer networks include
failure detectors \cite{AAIJR15} and the Byzantine agreement problem \cite{DFGKRT13}.
In \cite{AAIJR15}, the authors introduced a failure detector class suitable 
for a homonymous system prone to crash faults,
where processes are partially synchronous and links are eventually timely, 
and show how to implement and use it in such a system,
without assuming the number of processes as initial information.
In \cite{DFGKRT13}, the authors showed that the Byzantine agreement problem is solvable
if and only if $\ell \geq 3f + 1$ in the synchronous case,
and it is solvable only if $\ell > \frac{n + 3f}{2}$ in the partially synchronous case,
where $\ell$ is the number of distinct identifiers and $f$ is an upperbound on
the number of faulty processors.
Thus, the MAS of the Byzantine agreement problem is $3f+1$ in the synchronous case,
and it is greater than $\frac{n + 3f}{2}$ in the partially synchronous case.

The anonymous mobile robot model was introduced in 1996 \cite{SY96}.
The model was constructed, inspired by the development of distributed robot systems
such as multi-robot systems, drone networks, satellite constellations, and so on.
Unlike the anonymous computer network model,
system elements called robots reside in the Euclidean space,
and change their positions, responding to the current positions of other robots.
Some research works are cited in the first paragraph of this section.

The anonymous mobile agent model introduced in 2001 is another model of 
anonymous distributed system \cite{DFPS01}.
It models software agents which repeatedly migrate from a computer to another 
through communication links.
In the anonymous mobile agent model, anonymous agents move in 
a given anonymous finite graph.
The black-hole problem \cite{DFPS01} and the problem of searching for
an intruder \cite{BFFS02} were investigated in early days.

The anonymous mobile robot and the anonymous mobile agent models share many
problems such as the gathering and the pattern formation problems.
Solving such problems in a self-stabilizing manner while keeping the spirit of 
minimalist as much as possible is one of their main issues
(see surveys \cite{CSN19,Ilcinkas19,Luna19}).
For the literature of self-stabilization in general, see, e.g., \cite{Dolve00}.
Sometimes concepts and techniques introduced in the research 
of anonymous computer network (e.g., the view \cite{SY99,YK88}) 
are also applied to solve these problems.

There are a few papers that treat homonymous swarms, e.g., \cite{ASY22,AY23,CDFH11,LYKY18}.
Team assembly of heterogeneous robots, each dedicated to solve a subtask,
is discussed in \cite{LYKY18} (see also \cite{LKAV21}).
In \cite{ASY22}, robots' identifiers are used to specify and evaluate 
the quality of the robots' trajectories.
The compatibility of target functions is discussed \cite{AY23,CDFH11}:
A set $\Phi$ of target functions is compatible with respect to a problem $\Pi$,
if a swarm of robots always solves $\Pi$,
as long as each robot takes its target function from $\Phi$.
Thus, two swarms whose robots take target functions from $\Phi$ 
can freely merge to form a larger swarm that solves $\Pi$.

Since most of natural distributed systems consist of anonymous entities,
and do not have a central controller,
natural distributed systems have inspired researchers to introduce 
other anonymous distributed system models.
One of the models is the population protocol model,
which was introduced in 2006 \cite{AADFP06}.
It is a model for a collection of agents,
which are identically programmed finite state machines.
The amoebot model is another model introduced in 2014 \cite{DDGRS14}
(see also a survey \cite{DHRS19}).
Like the anonymous mobile agent model,
anonymous particles of the amoebot model move in a graph,
but unlike the anonymous mobile agent model,
they move in an anonymous infinite graph.

Finally, the concept of MAS can be introduced to each of the anonymous systems
after allowing processes, robots, agents or amoebots to have different identifiers.

\medskip
\noindent
{\bf Contributions.}
This paper investigates the MAS for various self-stabilizing gathering 
and related problems,
which are asked to solve problems from {\bf any} initial configuration.
Throughout the paper,
the scheduler is assumed to be {\bf semi-synchronous} ($\cal SSYNC$),
i.e., on a robot $r$, a Look-Compute-Move cycle starts 
at an integral time instant $t$,
and ends before (not including) $t+1$.
More carefully, $r$ observes the robots' positions at time $t$,
and has reached its destination before the cycle ends.

The $c$-{\em scattering problem} ($c$SCT) is the problem of forming a
configuration in which robots are distributed at least $c$ different positions.
The {\em scattering problem}, sometimes called the {\em split problem}, 
is the $n$SCT.
The $c$-{\em gathering problem} ($c$GAT) is the problem of forming a
configuration in which robots are distributed at most $c$ different positions.
The {\em gathering problem} (GAT) is thus 1GAT.
The {\em pattern formation problem} (PF) for a pattern $G$ is the problem of
forming a configuration $P$ similar\footnote{
Throughout this paper, 
we say that one object is similar to another,
if the latter is obtained from the former by a combination of
scaling, translation, and rotation (but not using a reflection).}
to $G$.

We also investigate problems in the presence of {\em crash failures}:
A faulty robot can stop functioning at any time,
becoming permanently inactive.
A faulty robot may not cause a malfunction, forever.
We cannot distinguish such a robot from a non-faulty one.

The {\em f-fault tolerant c-scattering problem} ($f$F$c$S) is
the problem of forming a configuration in which robots 
are distributed at $c$ (or more) different positions,
as long as at most $f$ robots have crashed.
The {\em f-fault tolerant gathering problem} ($f$FG) is
the problem of gathering {\bf all non-faulty robots} at a point,
as long as at most $f$ robots have crashed.
The {\em f-fault tolerant gathering problem to f points} ($f$FGP)
is the problem of gathering {\bf all robots} (including faulty ones)
at $f$ (or less) points, as long as at most $f$ robots have crashed.

Table~\ref{T0010} summarizes main results.

\begin{table}
\caption{
For each self-stabilizing problem $\Pi$,
the MAS for $\Pi$, an algorithm for $\Pi$ achieving the MAS
(and the theorem/corollary/observation citation number 
establishing the result in parentheses) are shown.
$c$SCT is the $c$-scattering problem,
$c$GAT is the $c$-gathering problem,
PF is the pattern formation problem,
$f$F$c$S is the fault tolerant $c$-scattering problem
in the presence of at most $f$ faulty robots,
$f$FG is the fault tolerant gathering problem
in the presence of at most $f$ faulty robots, and
$f$FGP is the fault tolerant problem of gathering 
all robots (including a faulty one) at $f$ (or less) points
in the presence of at most $f$ faulty robots.
$f$F$c$S (for some values of $f$ and $c$) and
$f$FGP are unsolvable, thus their MASs are $\infty$.}
\label{Table0010} 

\smallskip

\centering
\begin{tabular}{|c|c|c|}
\hline
problem $\Pi$                & MAS      & algorithm  \\ \hline \hline
$c$SCT ($1 \leq c \leq n$)   & $c$      & $c$SCTA (Thm.~\ref{T0010})\\
$c$GAT ($2 \leq c \leq n$)   & 1        & 2GATA (Cor.~\ref{C1010})\\
GAT (= 1GAT)                 & 2        & GATA (Thm.~\ref{T1010}) \\
PF                           & $n$      & PFA (Thm.~\ref{T3020})\\ 
$f$F1S ($1 \leq f \leq n-1$) & 1        & 1SCTA (Obs.~\ref{Ofs0010})\\
$f$F2S ($1 \leq f \leq n-2$) & $f+2$        & $(f+2)$SCTA (Thm.~\ref{Tfs0010})\\
$(n-1)$F2S                   & $\infty$ & --~~ (Thm.~\ref{Tfs0010})\\
$f$F$c$S ($c \geq 3$, $c+f-1 \leq n$) & $c+f-1$ & $(c+f-1)$SCTA (Thm.~\ref{Tfs0010})\\
$f$F$c$S ($c \geq 3$, $c+f-1 > n$) & $\infty$ & --~~ (Thm.~\ref{Tfs0010})\\
$f$FG  ($1 \leq f \leq n-1$) & 3        & SGTA (Thm.~\ref{T4040})\\
$f$FGP ($1 \leq f \leq n-1$) & $\infty$ & --~~ (Thms.~\ref{T4050},\ref{T4070})\\
\hline 
\end{tabular}

\end{table}

\medskip
\noindent
{\bf Organization.}
After introducing the robot model and several measures we will use in this paper 
in Section~\ref{Sintro},
we first establish the MAS of the $c$-scattering problem in Section~\ref{Sscattering}.
Then the MASs of the $c$-gathering and the pattern formation problems are
respectively investigated in Sections~\ref{Sgathering} and \ref{Spatternformation}.
Sections~\ref{Sfsct} and \ref{Sfgp} consider the MASs of the fault tolerant
scattering and the gathering problems, respectively.
Finally, 
we conclude the paper by giving several open problems in Section~\ref{Sconclusion}.

\section{Preliminaries}
\label{Sintro}

\subsection{The model}
\label{SSmodel}

Consider a swarm $\cal R$ of $n$ robots $r_1, r_2, \ldots , r_n$.
Each robot $r_i$ has its own unit of length and a local compass,
which define an $x$-$y$ local coordinate system $Z_i$:
$Z_i$ is right-handed, 
and the origin $(0,0)$ always shows the position of $r_i$, 
i.e., it is self-centric.
Robot $r_i$ has the strong multiplicity detection capability,
and can count the number of robots resides at a point.

A {\em target function} $\phi$ is a function from $(R^2)^n$ 
to $R^2 \cup \{ \bot \}$ for all $n \geq 1$ such that $\phi(P) = \bot$,
if and only if $(0,0) \not\in P$.
Here, $\bot$ is a special symbol to denote that $(0,0) \not\in P$;
since $Z_i$ is self-centric,
$(0,0) \not\in P$ means an error of eye sensor,
which we assume will not occur.
Given a target function $\phi_i$,
$r_i$ executes a Look-Compute-Move cycle when it is activated:
\begin{description}
 \item[Look:]  
$r_i$ identifies the multiset $P$ of the robots' positions in $Z_i$.

\item[Compute:] 
$r_i$ computes $\bm{x}_i = \phi_i(P)$.
(Since $(0,0) \in P$, $\phi_i(P) \not= \bot$.
In case $\phi_i$ is not computable,
we simply assume that $\phi_i(P)$ is given by an oracle.)

\item[Move:] 
$r_i$ moves to $\bm{x}_i$.
(We assume that $r_i$ always reaches $\bm{x}_i$ before this Move phase ends.)
\end{description}

We assume a discrete time $0, 1, \ldots$.
At each time $t \geq 0$, 
the scheduler nondeterministically activates some (possibly all) robots.
Then activated robots execute a cycle which starts at 
$t$ and ends before (not including) $t+1$,
i.e., the scheduler is semi-synchronous (${\cal SSYNC}$).

Let $Z_0$ be the $x$-$y$ global coordinate system.
It is right-handed.
The coordinate transformation from $Z_i$ to $Z_0$ is denoted by $\gamma_i$.
We use $Z_0$ and $\gamma_i$ just for the purpose of explanation.
They are not available to any robot $r_i$.

The position of robot $r_i$ at time $t$ in $Z_0$ is denoted by $\bm{x}_t(r_i)$.
Then $P_t = \{ \bm{x}_t(r_i) : 1 \leq i \leq n \}$ is a multiset representing
the positions of all robots at time $t$,
which is called the {\em configuration} of $\cal R$ at $t$.

Given an initial configuration $P_0$,
an assignment $\cal A$ of a target function $\phi_i$ to each robot $r_i$,
and an $\cal SSYNC$ schedule,\footnote{
An $\cal SSYNC$ schedule is an activation schedule produced by the $\cal SSYNC$ scheduler.}
the execution is a sequence 
${\cal E}: P_0, P_1, \ldots , P_t, \ldots$ of configurations starting from $P_0$.
Here, for all $r_i$ and $t \geq 0$,
if $r_i$ is not activated at $t$, then $\bm{x}_{t+1}(r_i) = \bm{x}_t(r_i)$.
Otherwise, if it is activated,
$r_i$ identifies $Q^{(i)}_t = \gamma^{-1}_i(P_t)$ in $Z_i$, 
computes $\bm{y} = \phi_i(Q^{(i)}_t)$, 
and moves to $\bm{y}$ in $Z_i$.
(Since $(0,0) \in Q^{(i)}_t$, $\bm{y} \not= \bot$.)
Then $\bm{x}_{t+1}(r_i) = \gamma_i(\bm{y})$.
We assume that the scheduler is fair:
It activates every robot infinitely many times.
Throughout the paper, 
we regard the scheduler as an adversary.

The ${\cal SSYNC}$ scheduler is said to be {\em fully synchronous} ($\cal FSYNC$),
if every robot $r_i$ is activated every time instant $t = 0, 1, 2, \ldots$.
The scheduler which is not $\cal SSYNC$ is said to be {\em asynchronous} ($\cal ASYNC$).
Throughout the paper, we assume that the scheduler is $\cal SSYNC$.

\subsection{Orders and symmetries}
\label{SSconcepts}

We use three orders $<$, $\sqsubset$, and $\succ$
(besides the conventional order $<$ on $R$).
Let $<$\footnote{
We use the same notation $<$ to denote the lexicographic order on $R^2$
and the order on $R$ to save the number of notations.}
be a lexicographic order on $R^2$.
For distinct points $\bm{p} = (p_x, p_y)$ and $\bm{q} = (q_x, q_y)$,
$\bm{p} < \bm{q}$, if and only if either (i) $p_x < q_x$, 
or (ii) $p_x = q_x$ and $p_y < q_y$ holds.
Let $\sqsubset$ be a lexicographic order on $(R^2)^n$.
For distinct multisets of $n$ points
$P = \{ \bm{p}_1, \bm{p}_2, \ldots, \bm{p}_n \}$
and $Q = \{ \bm{q}_1, \bm{q}_2, \ldots, \bm{q}_n \}$,
where for all $i = 1, 2, \ldots , n-1$, 
$\bm{p}_i \leq \bm{p}_{i+1}$ and $\bm{q}_i \leq \bm{q}_{i+1}$ hold,
$P \sqsubset Q$, if and only if there is an $i (1 \leq i \leq n-1)$
such that (i) $\bm{p}_j = \bm{q}_j$ for all $j = 1, 2, \ldots, i-1$,\footnote{
We assume $\bm{p}_0 = \bm{q}_0$.}
and (ii) $\bm{p}_i < \bm{q}_i$.

Let $P = \{\bm{p}_1, \bm{p}_2, \ldots , \bm{p}_n \} \in (R^2)^n$.
The set of distinct points of $P$ is denoted by 
$\overline{P} = \{\bm{q}_1, \bm{q}_2, \ldots , \bm{q}_m \}$,
where $|P| = n$ and $|\overline{P}| = m$.
We denote the multiplicity of $\bm{q}$ in $P$ by $\mu_P(\bm{q})$,
i.e., $\mu_P(\bm{q}) = |\{ i: \bm{p}_i = \bm{q} \in P \}|$.
We identify $P$ with a pair $(\overline{P},\mu_P)$,
where $\mu_P$ is a labeling function to associate 
label $\mu_P(\bm{q})$ with each element $\bm{q} \in \overline{P}$.

Let $G_P$ be the rotation group $G_{\overline{P}}$ of $\overline{P}$ 
about $\bm{o}_P$ preserving $\mu_P$,
where $\bm{o}_P$ is the center of the smallest enclosing circle of $P$.
The order $|G_P|$ of $G_P$ is denoted by $k_P$.
We assume that $k_P = 0$, 
if $|\overline{P}| = 1$, i.e., if $\overline{P} = \{ \bm{o}_P \}$.

The symmetricity of $P$ is $\sigma(P) = GCD(k_P, \mu_P(\bm{o}_P))$,
i.e., the greatest common divisor of $k_P$ and $\bm{o}_P$ \cite{SY99}.
Here, $GCD(\ell,0) = GCD(0, \ell) = \ell$.

\begin{example}
\label{Ex0010}
Let $P_1 = \{ \bm{a}, \bm{b}, \bm{c} \}$,
$P_2 = \{ \bm{a}, \bm{a}, \bm{b}, \bm{b}, \bm{c} \}$, and
$P_3 = \{ \bm{a}, \bm{a}, \bm{b}, \bm{b}, \bm{c}, \bm{c} \}$,
where a triangle $\bm{abc}$ is equilateral.
Then $k_{P_1} = \sigma(P_1) = 3$,
$k_{P_2} = \sigma(P_2) = 1$,
and $k_{P_3} = \sigma(P_3) = 3$.

Let $P_4 = \{ \bm{a}, \bm{b}, \bm{c}, \bm{o}, \bm{o} \}$, 
$P_5 = \{ \bm{a}, \bm{b}, \bm{c}, \bm{o}, \bm{o}, \bm{o} \}$, 
and $P_6 = \{ \bm{o}, \bm{o}, \bm{o} \}$, 
where $\bm{o}$ is the center of the smallest enclosing circle of triangle $\bm{abc}$.
Then $k_{P_4} = 3$, $\sigma(P_4) = 1$,
$k_{P_5} = \sigma(P_5) = 3$,
$k_{P_6} = 0$, and $\sigma(P_6) = 3$.
(See 
Figure~\ref{Fsym} for an illustration.)
\end{example}

\begin{figure}[t]
\centering
  \includegraphics[width=0.4\hsize]{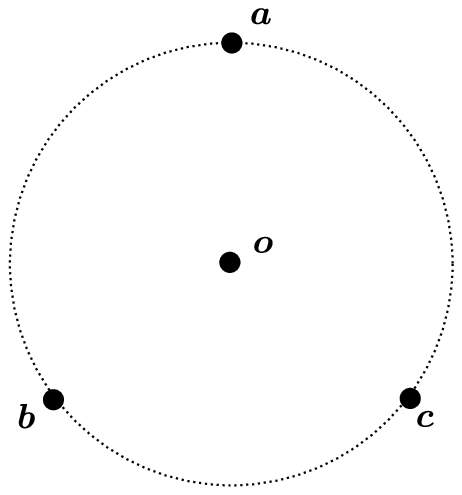}
 \caption{A configuration $P$, where $\overline{P} = \{\bm{a}, \bm{b}, \bm{c}, \bm{o}\}$.
A triangle $\bm{abc}$ is equilateral, 
and $\bm{o}$ is the center of the smallest enclosing circle of $P$.
If $\mu_P(\bm{a})=\mu_P(\bm{b})=\mu_P(\bm{c})= i$ for an integer $i > 0$, $k_P = 3$;
otherwise, $k_P = 1$.
If ($k_P = 3$ and) $\mu_P(\bm{o}) = 3j$ for an integer $j \geq 0$, then $\sigma(P)=3$;
otherwise, $\sigma(P)=1$.}
\label{Fsym}
\end{figure}

We use both $k_P$ and $\sigma(P)$.
Suppose that $P$ is a configuration in $Z_0$.
When activated, a robot $r_i$ identifies the robots' positions 
$Q^{(i)} = \gamma^{-1}_i(P)$ in $Z_i$ in Look phase.
Since $P$ and $Q^{(i)}$ are similar,
$k_P = k_{Q^{(i)}}$ and $\sigma(P) = \sigma(Q^{(i)})$,
i.e., all robots can consistently compute $k_P$ and $\sigma(P)$.

On the contrary, robots cannot consistently compute
lexicographic orders $<$ and $\sqsubset$.
To see this,
let $\bm{x}$ and $\bm{y}$ be distinct points in $\overline{P}$ in $Z_0$.
Then both $\gamma^{-1}_i(\bm{x}) < \gamma^{-1}_i(\bm{y})$
and $\gamma^{-1}_i(\bm{x}) > \gamma^{-1}_i(\bm{y})$ can occur,
depending on $Z_i$.
Thus robots may be unable to consistently compare $\bm{x}$ and $\bm{y}$ using $>$.
And it is true for $\sqsubset$, as well.
We thus introduce a total order $\succ_P$ on $\overline{P}$, 
in such a way that all robots can agree on the order,
provided $k_P = 1$.
A key trick behind the definition of $\succ_P$ is to use, instead of $Z_i$,
an $x$-$y$ coordinate system $\Xi_i$ which is computable 
for any robot $r_j$ from $Q^{(j)}$.

Let $\Gamma_P(\bm{q}) \subseteq \overline{P}$ be the orbit of $G_P$ 
through $\bm{q} \in \overline{P}$. 
Then $|\Gamma_P(\bm{q})| = k_P$ if $\bm{q} \not= \bm{o}_P$,
and $\mu_P(\bm{q}') = \mu_P(\bm{q})$ if and only if $\bm{q}' \in \Gamma_P(\bm{q})$.
If $\bm{o}_P \in \overline{P}$, $\Gamma_P(\bm{o}_P) = \{ \bm{o}_P \}$.
Let $\Gamma_P = \{ \Gamma_P(\bm{q}): \bm{q} \in \overline{P} \}$.
Then $\Gamma_P$ is a partition of $\overline{P}$.
Define $x$-$y$ coordinate system $\Xi_{\bm{q}}$ for any point 
$\bm{q} \in \overline{P} \setminus \{ \bm{o}_P \}$.
The origin of $\Xi_{\bm{q}}$ is $\bm{q}$,
the unit distance is the radius of the smallest enclosing circle of $P$,
the $x$-axis is taken so that it goes through $\bm{o}_P$,
and it is right-handed.
Let $\gamma_{\bm{q}}$ be the coordinate transformation from $\Xi_{\bm{q}}$ to $Z_0$.
Then the view $V_P(\bm{q})$ of $\bm{q}$ is defined to be $\gamma^{-1}_{\bm{q}}(P)$.
Obviously $V_P(\bm{q}') = V_P(\bm{q})$ (as multisets), 
if and only if $\bm{q}' \in \Gamma_P({\bm{q}})$.
Let $View_P = \{ V_P({\bm{q}}): \bm{q} \in \overline{P} \setminus \{ \bm{o}_P \} \}$.
(See 
Figure~\ref{Fxi} 
for an illustration.)

\begin{figure}[t]
\centering
 \includegraphics[width=0.4\hsize]{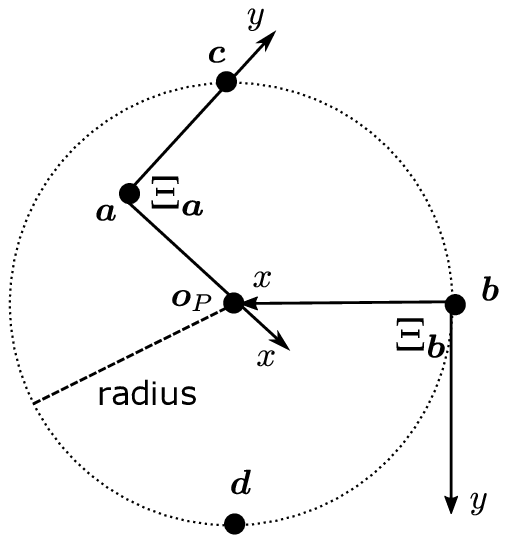}
 \caption{A configuration $P$, where $P = \overline{P} = \{\bm{o}_P, \bm{a}, \bm{b}, \bm{c}, \bm{d}\}$.
In $Z_0$, 
$\bm{o}_P=(0,0)$, $\bm{a} = (-1/2, 1/2)$, $\bm{b} = (1,0)$, $\bm{c} = (0,1)$, and $\bm{d}=(0,-1)$.
Then the smallest enclosing circle $C$ of $P$ has the center $\bm{o}_P$ and radius $1$.
Solid arrows represent the $x$- and $y$-axes of $\Xi_{\bm{a}}$ and $\Xi_{\bm{b}}$
with the unit length (the radius $1$ of $C$). 
In $\Xi_{\bm{b}}$, points $\bm{o}_P$, $\bm{a}$, $\bm{b}$, $\bm{c}$, and $\bm{d}$ are 
$(1,0)$, $(3/2,-1/2)$, $(0,0)$, $(1,-1)$, and $(1,1)$, respectively,
and thus $\gamma^{-1}_{\bm{b}}(P) = V_P(\bm{b}) = \{(1,0), (3/2, -1/2), (0,0), (1, -1), (1,1)\}$. }
\label{Fxi}
\end{figure}

Any robot $r_i$, in Compute phase,
can compute $\Xi_{\bm{q}}$ and $V_{Q^{(i)}}(\bm{q})$
for each $\bm{q} \in \overline{Q^{(i)}} \setminus \{ \bm{o}_{Q^{(i)}} \}$,
and thus $View_{Q^{(i)}}$, from $Q^{(i)}$.
Since $P$ and $Q^{(i)}$ are similar,
by the definition of $\Xi_{\bm{q}}$, $View_{P} = View_{Q^{(i)}}$,
which implies that all robots $r_i$ can consistently compute $View_{P}$.

We define $\succ_P$ on $\Gamma_P$ using $View_P$.
For any distinct orbits $\Gamma_P(\bm{q})$ and $\Gamma_P(\bm{q}')$ in $\Gamma_P$,
$\Gamma_P(\bm{q}) \succ_P \Gamma_P(\bm{q}')$, if and only if
one of the following conditions hold:
\begin{enumerate}
\item 
$\mu_P(\bm{q}) > \mu_P(\bm{q}')$.

\item
$\mu_P(\bm{q}) = \mu_P(\bm{q}')$ 
and $dist(\bm{q}, \bm{o}_P) < dist(\bm{q}', \bm{o}_P)$ hold,
where $dist(\bm{x}, \bm{y})$ is the Euclidean distance 
between $\bm{x}$ and $\bm{y}$.

\item
$\mu_P(\bm{q}) = \mu_P(\bm{q}')$,
$dist(\bm{q}, \bm{o}_P) = dist(\bm{q}', \bm{o}_P)$, and
$V_P(\bm{q}) \sqsupset V_P(\bm{q}')$ hold.
\footnote{Since $dist(\bm{o}_P, \bm{o}_P) = 0$,
$V_P({\bm{q}})$ is not compared with $V_P(\bm{o}_P)$ with respect to $\sqsupset$.}
\end{enumerate}
Then $\succ_P$ is a total order on $\Gamma_P$.
If $k_P = 1$, since $\Gamma_P(\bm{q}) = \{ \bm{q} \}$ 
for all $\bm{q} \in \overline{P}$,
we regard $\succ_P$ as a total order on $\overline{P}$
by identifying $\Gamma_P({\bm{q}})$ with $\bm{q}$.
For a configuration $P$ (in $Z_0$), from $Q^{(i)}$ (in $Z_i$),
each robot $r_i$ can consistently compute $k_P = k_{Q^{(i)}}$,
$\Gamma_P = \Gamma_{Q^{(i)}}$, and $View_P = View_{Q^{(i)}}$, 
and hence $\succ_P = \succ_{Q^{(i)}}$.
Thus, all robots can agree on, e.g., 
the largest point $\bm{q} \in \overline{P}$ with respect to $\succ_P$,
provided $k_P = 1$.

Let $\mathcal{P} = \{ P \in (R^2)^n: (0,0) \in P,\  n \geq 1 \}$.
Since a target function returns $\bot$ when $(0,0) \not\in P$,
we regard $\mathcal{P}$ as the domain of a target function.
For two sets $X, Y (\subseteq R^2)$ of points,
$dist(X,Y) = \min_{\bm{x} \in X, \bm{y} \in Y} dist(\bm{x},\bm{y})$,
and $dist(\bm{x},Y) = dist(\{\bm{x}\},Y)$.
For a configuration $P$, 
let $-P = \{ -\bm{p}: \bm{p} \in P\}$.

\section{$C$-scattering problem}
\label{Sscattering}

We start with showing that there is a problem whose MAS is $c$ for all integer $c > 0$.
Thus the problems form an infinite hierarchy with respect to their MASs.\footnote{
We will show the existence of a problem whose MAS is $\infty$.
However, this claim holds, 
regardless of whether or not there is such a problem with MAS being $\infty$.}

The {\em scattering problem} (SCT), sometimes called the {\em split problem},
is the problem to have the robots occupy distinct positions,
starting from any configuration \cite{DP09}.
For $1 \leq c \leq n$,
let the {\em c-scattering problem} ($c$SCT) be the problem 
to transform any initial configuration to a configuration $P$ 
such that $|\overline{P}| \geq c$.
Thus, the $n$SCT is the SCT.
An algorithm for the $c$SCT is an algorithm for the $(c-1)$SCT,
for $2 \leq c \leq n$.

\begin{theorem}
 \label{T0010}
For any $1 \leq c \leq n$,
the MAS for the \mbox{\rm $c$SCT} is $c$.
\end{theorem}

\begin{proof}
(I) We first show that the MAS for the $c$SCT is at least $c$.
Proof is by contradiction.
Suppose that the MAS for the $c$SCT is $m < c$ to derive a contradiction.
Let $\Phi = \{ \phi_1, \phi_2, \ldots , \phi_m \}$ be an algorithm 
for the $c$SCT.
Consider the following situation:
\begin{enumerate}
\item 
All robots $r_i$ ($1 \leq i \leq n$) share the unit length 
and the direction of positive $x$-axis.
\item
A target function assignment $\cal A$ is defined as follows:
${\cal A}(r_i) = \phi_i$ for $1 \leq i \leq m-1$,
and ${\cal A}(r_i) = \phi_m$ for $m \leq i \leq n$.
\item
All robots initially occupy the same location $(0,0)$.
That is,  $P_0 = \{ (0,0),$ $(0,0), \ldots , (0,0) \}$.
\item
The scheduler is ${\cal FSYNC}$.
\end{enumerate}

Let ${\cal E}: P_0, P_1, \ldots$ be the execution of $\cal R$ starting from $P_0$,
under the above situation.
By an easy induction on $t$,
all robots $r_i$ ($m \leq i \leq n$) occupy the same location, i.e., for all $t \geq 0$,
$\bm{x}_t(r_m) = \bm{x}_t(r_{m+1}) = \ldots = \bm{x}_t(r_n)$.
Since $|\overline{P_t}| \leq m < c$ for all $t \geq 0$,
a contradiction is derived.

\medskip
\noindent

(II) We next present an algorithm $c$SCTA = $\{ sct_1, sct_2, \ldots, sct_c \}$
of size $c$ for the $c$SCT,
where target function $sct_i$ is defined as follows for any $P \in \mathcal{P}$.

\medskip
\noindent
{\bf [Target function $sct_i$]}
\vspace{-1mm}
\begin{enumerate}
\item
If $|\overline{P}| \geq c$, then $sct_i(P) = (0,0)$ for $i = 1, 2, \ldots , c$.

\item
If $|\overline{P}| = 1$, 
then $sct_1(P) = (0,0)$, and $sct_i(P) = (1,0)$ for $i = 2, 3, \ldots , c$.

\item
If $2 \leq |\overline{P}| \leq c-1$, 
$sct_i(P) = (\delta/(2(i+1)),0)$ for $i = 1, 2, \ldots , c$,
where $\delta$ is the smallest distance between two (distinct) 
points in $\overline{P}$.
\end{enumerate}

We show that $c$SCTA is an algorithm for the $c$SCT.
Consider any execution ${\cal E}: P_0, P_1, \ldots ,$
starting from any initial configuration $P_0$.
We can assume $|\overline{P_0}| < c$ without loss of generality,
since $|\overline{P_t}| \geq c$ implies $P_t = P_{t'}$ for all $t' \geq t$.

Proof is by contradiction.
To derive a contradiction,
we assume that $|\overline{P_t}| < c$ for all $t \geq 0$.
Let $m = \max_{t \geq 0} |\overline{P_t}| < c$,
and assume without loss of generality $|\overline{P_0}| = m$.

Suppose that $m = 1$,
i.e., $\overline{P_0} = \{ \bm{p} \}$ for some $\bm{p} \in R^2$.
By the fairness of the $\cal SSYNC$ scheduler,
let $t \geq 0$ be the time instant at which a robot $r_i$ that takes 
a target function $sct_a (a \not= 1)$ is activated for the first time.
Then $P_t = P_0$, and there is a robot $r_j$ that takes $sct_1$ at $\bm{p}$.
By the definition of $c$SCTA,
$\bm{x}_{t+1}(r_j) = \bm{p}$,
and $\bm{x}_{t+1}(r_i) = \gamma_i((1,0))$.
Thus $|\overline{P_{t+1}}| \geq 2$,
since $\bm{p} = \gamma_i((0,0)) \not= \gamma_i((1,0))$.
It is a contradiction.

Suppose that $1 < m~(< c)$.
First observe that $|\overline{P_t}| \geq m$ for all $t \geq 0$.
Suppose $\bm{x}_t(r_i) \not= \bm{x}_t(r_j)$.
Then $dist(\bm{x}_t(r_i), \bm{x}_t(r_j)) \geq \delta$.
Since they both move at most distance $\delta/4$,
$\bm{x}_{t+1}(r_i) \not= \bm{x}_{t+1}(r_j)$,
regardless of their local $x$-$y$ coordinate systems, $Z_i$ and $Z_j$.

Suppose that $|\overline{P_t}| = m$ holds for all $t \geq 0$.
At each $t$, 
if a robot $r$ at $\bm{x}_t(r) = \bm{q}$ is activated,
all robots $r'$ such that $\bm{x}_t(r') = \bm{q}$ must be activated simultaneously,
and move to the same location, i.e., their target functions are the same.
However, it is a contradiction.
There is a point $\bm{q} \in \overline{P_0}$ such that
$\bm{x}_0(r_i) = \bm{x}_0(r_j) = \bm{q}$ for some $i \not= j$, 
and $r_i$ and $r_j$ take different target functions,
since $m < c$.
Thus, there is a time $t' > t$ such that $|\overline{P_{t'}}| > |\overline{P_t}| = m$.

Therefore, $\cal E$ eventually reaches a configuration $P_t$
such that $|\overline{P_t}| \geq c$.
\end{proof}

\begin{corollary}
\label{C0010} 
Let $\cal R$ be a swarm consisting of $n$ robots.
For $c = 1, 2, \ldots , n$,
there is a problem $\Pi_c$ for $\cal R$
such that the MAS for $\Pi_c$ is $c$.
\end{corollary}

\begin{proof}
Let $\Pi_c$ be the $c$SCT.
Then the MAS for $\Pi_c$ is $c$.
\end{proof}

\section{$C$-gathering problem}
\label{Sgathering}

Let $P = \{\bm{p}_1, \bm{p}_2, \ldots , \bm{p}_n \} \in {\cal P}$,
$\overline{P} = \{\bm{q}_1, \bm{q}_2, \ldots , \bm{q}_{m_P} \}$,
$m_P = |\overline{P}|$ be the size of $\overline{P}$,
$\mu_P(\bm{q})$ denote the multiplicity of $\bm{q}$ in $P$,
$\bm{o}_P$ be the center of the smallest enclosing circle $C_P$ of $P$,
and $CH(P)$ be the convex hull of $P$.

The {\em c-gathering problem} ($c$GAT) is the problem of transforming 
any initial configuration to a configuration $P$
such that $|\overline{P}| \leq c$.
The 1GAT is thus the {\em gathering problem} (GAT).
An algorithm for the $c$GAT is an algorithm for the $(c+1)$GAT,
for $1 \leq c \leq n-1$.

Under the $\cal SSYNC$ scheduler,
the GAT from distinct initial positions is solvable (by an algorithm of size 1), 
if and only if $n \geq 3$~\cite{SY99},
and the GAT from any initial configuration is solvable (by an algorithm of size 1), 
if and only if $n$ is odd~\cite{DP12}.
The MAS for the GAT is thus at least 2.
Gathering algorithms $\psi_{f-point(n)}$ (for $n \geq 3$ robots 
from distinct initial positions) in Theorem~3.4 of \cite{SY99}
and Algorithm~1 (for odd $n$ robots from any initial positions) in \cite{DP12}
share the skeleton:
Given a configuration $P$,
if there is the (unique) ``largest point'' $\bm{q} \in \overline{P}$, then go to $\bm{q}$; 
otherwise, go to $\bm{o}_P$.
Consider the following singleton 2GATA $= \{ 2gat \}$ of a target function $2gat$,
which is a direct implementation of this strategy using $\succ_P$
as the measure to determine the largest point in $\overline{P}$.

\medskip
\noindent
{\bf [Target function $2gat$]}
\vspace{-1mm}
\begin{enumerate}
\item
If $m_P = 1$, or $m_P = 2$ and $k_P = 2$, i.e., $\mu_P(\bm{q}_1) = \mu_P(\bm{q}_2)$,
then $2gat(P) = (0,0)$.

\item
If $m_P \geq 2$ and $k_P = 1$, 
then $2gat(P) = \bm{q}$,
where $\bm{q} \in \overline{P}$ is the largest point 
with respect to $\succ_P$.

\item
If $m_P \geq 3$ and $k_P \geq 2$,
then $2gat(P) = \bm{o}_P$.
\end{enumerate}

A multiset $P$ is said to be {\em linear}
if $CH(P)$ is a line segment.
We need the following technical lemma.

\begin{lemma}[\cite{AY23}]
\label{L1005} 
Let $A$ be a set (not a multiset) of points satisfying 
(1) $A$ is not linear,
(2) $k_A \geq 2$, 
and (3) $\bm{o}_A \not\in A$.
For a point $\bm{a} \in A$,
let $B = (A \cup \{ \bm{o}_A \}) \setminus \{ \bm{a} \}$,
i.e., $B$ is constructed from $A$ by replacing $\bm{a} \in A$ with $\bm{o}_A$.
Then $k_B = 1$.
\end{lemma}

We call a configuration $P$ {\em unfavorable} if $m_P = k_P = 2$,
which is said to be {\em bivalent} in \cite{BDT13}.
Otherwise, it is {\em favorable}.

\begin{theorem}
\label{T1000} 
Algorithm 2GATA transforms any initial configuration $P_0$
either to a configuration $P$ satisfying $m_P = 1$,
or to an unfavorable configuration.
\end{theorem}

\begin{proof}
Consider any execution ${\cal E}: P_0, P_1, \ldots$ of 2GATA, 
starting from any initial configuration $P_0$.
Let $m_t =  |\overline{P_t}|$, $k_t = k_{P_t}$, $\mu_t = \mu_{P_t}$,
$C_t = C_{P_t}$, $\bm{o}_t = \bm{o}_{P_t}$, $CH_t = CH(P_t)$,
and $\succ_t = \succ_{P_t}$.

We first claim that, 
there is a $t \geq 0$ such that $m_t = 1$ if $k_0 = 1$.

Let $\bm{q} \in \overline{P_0}$ be the largest point with respect to $\succ_{0}$.
It suffices to show that $\mu_t(\bm{q}) < \mu_{t+1}(\bm{q})$ for all $t \geq 0$,
provided that $\mu_t(\bm{q}) < n$.
The proof is by induction on $t$ for all $t \geq 1$.
The induction hypothesis $IH(t)$ is that
(1) $\mu_t(\bm{q}) > \mu_t(\bm{q}')$ for any $\bm{q}' (\not= \bm{q}) \in \overline{P_t}$ 
(i.e., $\bm{q}$ is hence the largest point with respect to $\succ_{t}$),
(2) $\bm{q} = \bm{o}_t$ holds if $k_t > 1$, and
(3) $\mu_t(\bm{q}) > \mu_{t-1}(\bm{q})$.

As for the base case,
when $t = 1$, since $k_0 = 1$,
robots activated at time $0$ move to $\bm{q} \in \overline{P_0}$.
Since $\mu_0(\bm{q}) \geq \mu_0(\bm{q}')$ for any $\bm{q}' \in \overline{P_0}$ 
by the definition of $\succ_0$,
$\mu_{1}(\bm{q}) > \mu_0(\bm{q}) \geq \mu_0(\bm{q}') \geq \mu_{1}(\bm{q}')$.
Thus $\mu_{1}(\bm{q}) > \mu_{1}(\bm{q}')$,
$\bm{q} = \bm{o}_1$ if $k_1 > 1$,
and $\mu_1(\bm{q}) > \mu_0(\bm{q})$.
Thus $IH(1)$ holds.

As for the induction step, for $t = 1, 2, \ldots,$
we show $IH(t+1)$, provided $IH(t)$,
which however is almost the same as the base case.
When $k_t = 1$,
since $\mu_t(\bm{q}) > \mu_t(\bm{q}')$ for any $\bm{q}' (\not= \bm{q}) \in \overline{P_t}$,
and hence $\bm{q}$ is hence the largest point with respect to $\succ_{t}$,
robots activated at time $t$ move to $\bm{q}$,
as in the base case,
we have $\mu_{t+1}(\bm{q}) > \mu_t(\bm{q}) > \mu_t(\bm{q}') \geq \mu_{t+1}(\bm{q}')$,
which implies that
$\mu_{t+1}(\bm{q}) > \mu_{t+1}(\bm{q}')$,
$\bm{q} = \bm{o}_{t+1}$ if $k_{t+1} > 1$,
and $\mu_{t+1}(\bm{q}) > \mu_t(\bm{q})$,
i.e., $IH(t+1)$ holds.
Otherwise, if $k_t \geq 2$, 
robots activated at time $t$ move to $\bm{o}_t$,
which is $\bm{q}$ by $IH(t)$.
Thus, $\mu_{t+1}(\bm{q}) > \mu_{t+1}(\bm{q}')$,
$\bm{q} = \bm{o}_{t+1}$ if $k_{t+1} > 1$,
and $\mu_{t+1}(\bm{q}) > \mu_t(\bm{q})$,
i.e., $IH(t+1)$ holds.

We can thus assume that $k_t \geq 2$ for all $t \geq 0$ without loss of generality.
We next show that if $P_t$ is favorable, so is $P_{t+1}$,
i.e., $P_t$ is favorable for all $t \geq 0$ (provided that $P_0$ is favorable).
Proof is by contradiction.
Suppose that $P_t$ is favorable, but $P_{t+1}$ is not, to derive a contradiction.
By assumption, $m_t \geq 3$ and $k_t \geq 2$.
Let $\overline{P_{t}} = \{ \bm{q}_1, \bm{q}_2, \ldots, \bm{q}_{m_t}\}$ 
and $\overline{P_{t+1}} = \{ \bm{p}_1, \bm{p}_2 \}$.
Without loss of generality, $\bm{p}_1 = \bm{o}_t$ and $\bm{p}_2 = \bm{q}_j$
for some $1 \leq j \leq m_t$,
since all robots activated at $t$ move to $\bm{o}_t$.
Note that $\bm{o}_t$ may not be in $\overline{P_t}$.
Observe that 
\[
\mu_{t+1}(\bm{p}_1) = n - \mu_{t+1}(\bm{p}_2) = n - \mu_{t+1}(\bm{q}_j)
\geq n - \mu_{t}(\bm{q}_j) > n/2,
\]
since $\mu_{t}(\bm{q}_j) < n/2$ (because $m_t \geq 3$ and $k_t \geq 2$),
which is a contradiction to the fact that $P_{t+1}$ is unfavorable.

Thus, in the rest of the proof,
we show that there is a time $t > 0$ such that $m_0 > m_t$,
provided that $P_t$ is favorable and $k_t \geq 2$ for all $t \geq 0$.
It suffices to show that there is a time $t > 0$ such that $m_t = 1$.

\medskip
\noindent
(I) First consider the case in which $P_0$ is linear.
Let $\overline{P_0} = \{ \bm{q}_1, \bm{q}_2, \ldots , \bm{q}_{m_0} \}$,
where $\bm{q}_i$ ($1 < i < m_0$) appear in $\overline{\bm{q}_1\bm{q}_{m_0}}$ 
in this order.
By the definition of 2GATA,
since $P_0$ is linear, 
so is $P_t$ for all $t \geq 0$,
and hence $k_t = 2$ since $k_t \geq 2$.

\medskip
\noindent
(I-A) Suppose that $m_0 (\geq 3)$ is odd.
Since $k_0 = 2$,  
$\bm{o}_0 = \bm{q}_{\lceil m_0/2 \rceil} \in \overline{P_0}$,
and all robots activated at time 0 move to $\bm{o}_0$.
Since $m_1 \leq m_0$, $m_1 = m_0$ (because we have nothing to show if $m_1 < m_0$),
$\overline{P_1} = \overline{P_0}$, $\mu_1(\bm{o}_1) = \mu_1(\bm{o}_0) > \mu_0(\bm{o}_0)$,
and $k_1 = 2$ by assumption.
Repeating this argument,
since $\mu_t(\bm{o}_0) = n$ implies $m_t = 1$,
there is a time $t$ such that $m_t < m_0$.

\medskip
\noindent
(I-B) Suppose that $m_0 (\geq 4)$ is even.
Since $k_0 = 2$, 
all robots activated at time 0 move to 
$\bm{o}_0 = (\bm{q}_1 + \bm{q}_{m_0})/2 \not\in \overline{P_0}$.
Thus $m_0 \leq m_1 \leq m_0 + 1$ (since we have nothing to show if $m_1 < m_0$).

If $m_1 = m_0 + 1$, since $m_1$ is odd and $k_1 = 2$, by the argument in (I-A),
eventually $m_1 > m_t$ holds for some $t > 1$ for the first time.
Since all robots activated move to $\bm{o}_0$,
$\overline{P_{t'}} = \overline{P_1}$ for all $1 \leq t' < t$.
Since $k_t = 2$ and $m_t = m_0$ (since we have nothing to show if $m_t < m_0$),
$\overline{P_t} = \overline{P_1} \setminus \{ \bm{q}_e \} = 
(\overline{P_0} \cup \{ \bm{o}_0 \}) \setminus \{ \bm{q}_e \}$,
where $e$ is either $1$ or $m_0$.
Furthermore, since $k_0 (= k_1) = k_t = 2$,
for all $i = 1, 2, \ldots , m_0 - 1$ except for $m_0/2$,
$dist(\bm{q}_i, \bm{q}_{i+1}) = d$, 
and $dist(\bm{q}_{m_0/2},\bm{o}_0) = dist(\bm{o}_0,\bm{q}_{m_0/2+1}) = d$,
for some $d > 0$.
That is, eventually $\cal E$ reaches a configuration $P_t$
such that $k_t = 2$, $m_t = m_0$ is even, and all points in 
$\overline{P_t} (= (\overline{P_0} \cup \{ \bm{o}_0 \}) \setminus \{ \bm{q}_e \})$
occur in an equidistant manner in $CH_t$,
where $e$ is either $1$ or $m_0$.

If $m_1 = m_0$, since $k_1 = 2$,
by the same argument,
$k_1 = 2$, $m_1 = m_0$ is even, and all points in
$\overline{P_1} (= (\overline{P_0} \cup \{ \bm{o}_0 \}) \setminus \{ \bm{q}_e \})$
occur in an equidistant manner in $CH_1$,
where $e$ is either $1$ or $m_0$.

Thus, without loss of generality, we may assume that
$k_0 = 2$, $m_0$ is even, 
and all points $\bm{q}_i \in \overline{P_0}$
occur in an equidistant manner in $CH_0$.

We repeat the above arguments:
Since $k_{0} = k_{1} = 2$, 
it must hold 
$2 \cdot dist(\bm{q}_1,\bm{q}_2) =  dist(\bm{q}_{m_0/2},\bm{q}_{m_0/2+1})$ in $P_0$, 
which contradicts to the assumption for $P_0$.

\medskip
\noindent
(II) Next consider the case in which $P_0$ is not linear.
As in (I), $m_0 \leq m_1 \leq m_0 + 1$ 
(since we have nothing to show if $m_1 < m_0$).

\medskip
\noindent
(II-A) Suppose first that $m_1 = m_0 + 1$.
That is, $\bm{o}_0 \not\in \overline{P_0}$,
and $\overline{P_1} = \overline{P_0} \cup \{ \bm{o}_0 \}$.
Since $k_1 \geq 2$ and $\bm{o}_1 = \bm{o}_0$,
all robots activated at time $t$ move to $\bm{o}_0$,
as long as $\overline{P_t} = \overline{P_1}$ holds.

Let $t > 1$ be the first time such that $m_0 = m_t < m_{t-1} = m_1$ hold
(since we have nothing to show if $m_t < m_0$).
Since $\overline{P_{t-1}} = \overline{P_1} = \overline{P_0} \cup \{ \bm{o}_0 \}$, 
$\bm{o}_{t-1} = \bm{o}_0$, and $k_{t-1} \geq 2$,
all robots activated at $t-1$ move to $\bm{o}_0$.
Suppose that $\overline{P_t} = \overline{P_1} \setminus \{ \bm{q} \}$,
for some $\bm{q} \in \overline{P_0}$ (since $\bm{o}_0$ remains in $\overline{P_t}$).
Since $k_0 \geq 2$, $P_0$ is not linear, $\bm{o}_0 \not\in \overline{P_0}$,
and $\overline{P_{t}} = (\overline{P_0}\setminus\{\bm{q}\}) \cup \{ \bm{o}_0 \}$,
we have $k_{t} = 1$ by Lemma~\ref{L1005}. 
Thus eventually $\cal E$ reaches a configuration $P_{t'}$ such that $m_{t'} = 1$.

\medskip
\noindent
(II-B) Suppose next that $m_1 = m_0$.
There are two cases to be considered.
If $\bm{o}_0 \in \overline{P_0}$ and $\overline{P_1} = \overline{P_0}$,
all robots activated at $t$ move to $\bm{o}_0$, 
as long as $\overline{P_t} = \overline{P_0}$, since $k_0 \geq 2$.
Thus eventually $\cal E$ reaches a configuration $P_t$ such that $m_t < m_0$.

Otherwise, if $\bm{o}_0 \not\in \overline{P_0}$,
and $\overline{P_1} = (\overline{P_0} \cup \{ \bm{o}_0 \}) \setminus \{ \bm{q} \}$ 
for some $\bm{q} \in \overline{P_0}$, where $\bm{q}\neq \bm{o}_0$.
Since $k_0 \geq 2$, $P_0$ is not linear, $\bm{o}_0 \not\in \overline{P_0}$,
we have $k_{t} = 1$ by Lemma~\ref{L1005}. 
Thus eventually $\cal E$ reaches a configuration $P_{t'}$ such that $m_{t'} = 1$.
\end{proof}

\begin{corollary}
\label{C1010} 
The MAS for the \mbox{\rm $c$GAT} is 1, for all $2 \leq c \leq n$.
\end{corollary}

\begin{proof}
By Theorem~\ref{T1000},
2GATA is an algorithm for the 2GAT,
which implies that it is an algorithm for the $c$GAT
for all $2 \leq c \leq n$.
\end{proof}

\medskip

If we consider that a problem with a smaller MAS is easier than 
a one with a larger MAS, 
the $c$GAT is not harder than the $c$SCT for all $2 \leq c \leq n$.

\begin{corollary}
\label{C1030} 
\mbox{\rm 2GATA} solves the \mbox{\rm GAT},
if and only if the initial configuration is favorable.
\end{corollary}

\begin{proof}
It is immediate from the proof of Theorem~\ref{T1000}.
\end{proof}

\medskip

Corollary~\ref{C1030} has been obtained by some researchers:
The \mbox{\rm GAT} is solvable, 
if and only if $n$ is odd \cite{DP12}.
Or more precisely,
it is solvable,
if and only if the initial configuration is favorable 
(i.e., not bivalent) \cite{BDT13}.
Note that the algorithm of \cite{BDT13} makes use of the Weber point
and tolerates at most $n-1$ crashes.

Consider a set GATA $= \{ gat_1, gat_2 \}$ of 
target functions $gat_1$ and $gat_2$ defined as follows:

\medskip
\noindent
{\bf [Target function $gat_1$]}
\vspace{-1mm}
\begin{enumerate}
\item
If $m = 1$, then $gat_1(P) = (0,0)$.

\item
If $m = 2$ and $k_P = 1$, or $m \geq 3$,
$gat_1(P) = 2gat(P)$.

\item
If $m = 2$ and $k_P = 2$,
then $gat_1(P) = (0,0)$.
\end{enumerate}

\medskip
\noindent
{\bf [Target function $gat_2$]}
\vspace{-1mm}
\begin{enumerate}
\item
If $m = 1$, then $gat_2(P) = (0,0)$.

\item
If $m = 2$ and $k_P = 1$, or $m \geq 3$,
then $gat_2(P) = 2gat(P)$.

\item
Suppose that $m = 2$ and $k_P = 2$.
Let $\bm{q}$ be a point in $\overline{P}$ such that $\bm{q}\neq (0,0)$.
Since $(0,0) \in \overline{P}$,
$\bm{q}$ is uniquely determined.
If $\bm{q} > (0,0)$, then $gat_2(P) = \bm{q}$.
Else if $\bm{q} < (0,0)$, then $gat_2(P) = 2 \bm{q}$.
\end{enumerate}

\begin{theorem}
\label{T1010}
\mbox{\rm GATA} is an algorithm for the \mbox{\rm GAT}.
The MAS for the \mbox{\rm GAT} is, hence, $2$.
\end{theorem}

\begin{proof}
We show that GATA = $\{ gat_1, gat_2 \}$ is a gathering algorithm.
Consider any execution ${\cal E}: P_0, P_1, \ldots$ 
starting from any initial configuration $P_0$.
By Corollary~\ref{C1030},
we can assume that $P_0$ is unfavorable by the definition of GATA.
Without loss of generality, 
we may assume that $P_0$ satisfies 
$\overline{P_0} = \{ (0,0), (1,0)\}$ in $Z_0$,
$\mu_0((0,0)) = \mu_0((1,0)) = h$, i.e., $n = 2h$ is even.
By the definition of GATA again, 
it is easy to observe that $P_t$ is linear.
We show that there is a time $t$ such that $P_t$ is favorable.
Proof is by contradiction.

A $gat_1$ robot (i.e., a robot taking $gat_1$ as its target function) does not move,
if a configuration is unfavorable.
Since there is at least one $gat_1$ robot,
without loss of generality,
we assume that there is a $gat_1$ robot at $(0,0)$ at $t = 0$.
Since $P_t$ is unfavorable for all $t > 0$,
all $gat_1$ robots do not move forever.
Let $A_t(\bm{q})$ be the set of $gat_2$ robots at $\bm{q}$ activated at time $t$,
and let $a_t(\bm{q}) = |A_t(\bm{q})|$.

Consider any robot $r_i \in A_0((0,0))$.
If $(0,0) > \bm{q} ~(\not= (0,0))$ in $Z_i$,
then $(0,0), (2,0) \in \overline{P_1}$ by the definition of $gat_2$.
Since $P_1$ is unfavorable, 
$\overline{P_1} = \{ (0,0), (2,0) \}$,
which implies that all robots at $(1,0)$ move to $(0,0)$ by the definition of $gat_2$.
It is a contradiction, 
since $k_1 = 1$ because $\mu_1((0,0)) > h$.

Therefore, every robot $r_i \in A_0((0,0))$ satisfies $(0,0) < \bm{q}$ in $Z_i$,
and move to $(1,0)$,
which implies $(0,0), (1,0) \in \overline{P_1}$.
Since $P_1$ is unfavorable,
the following three conditions hold:
$\overline{P_1} = \{ (0,0), (1,0) \}$,
$a_0((1,0)) = a_0((0,0))$,
and all robots $r_j \in A_0((1,0))$ move to $(0,0)$,
each satisfying $(0,0) < \bm{q}$ in $Z_j$.
Then $P_0 = P_1$ holds.

We make a key observation here.
At $t = 1$, 
each robot $r_i \in A_0((0,0))$ (resp. $r_j \in A_0((1,0))$)
currently at $(1,0)$ (resp. $(0,0)$) satisfies
$(0,0) > \bm{q}$ in $Z_i$ (resp. in $Z_j$).
Therefore, any robot $r_i$ in $A_0((0,0))$ 
(resp.  $r_j$ in $A_0((1,0))$) is not included 
in $A_t((1,0))$ (resp. $A_t((0,0))$) as a member, 
as long as $P_0 = P_t$ holds.
Thus there is a time $t > 0$ such that $a_t((0,0)) = 0$ holds.
We consider the smallest such $t$.
That is, $a_{t'}((0,0)) > 0$ for $0 \leq t' \leq t-1$,
and $P_0 = P_t$ hold.

If there is a robot $r_j \in A_t((1,0))$
satisfying $(0,0) < \bm{q}$ in $Z_j$,
then $a_t((0,0)) > 0$, which is a contradiction.
Thus every robot $r_j \in A_t((1,0))$ satisfies $(0,0) > \bm{q}$ in $Z_j$,
and moves to $(-1,0)$, 
which implies that $(-1,0), (0,0) \in \overline{P_{t+1}}$.
Since $P_{t+1}$ is unfavorable, 
$\overline{P_{t+1}} = \{ (-1,0), (0,0) \}$,
and all robots $r_j$ at $(1,0)$ at time $t$ are activated at $t$,
and move to $(-1,0)$, each satisfying $(0,0) > \bm{q}$ in $Z_j$.
Since $P_{t+1}$ is unfavorable, $a_t((1,0)) = h$ and hence $t = 0$.

Without loss of generality,
we may assume that $\overline{P_0} = \{ (-1,0), (0,0) \}$,
$\mu_0((-1,$ $0)) = \mu_0((0,0)) = h$,
all robots $r_j$ at $(-1,0)$ are $gat_2$ robots,
each of which satisfies $(0,0) < \bm{q}$ in $Z_j$.

By the same argument above, 
every robot $r_i \in A_t((0,0))$ satisfies $(0,0) < \bm{q}$ in $Z_i$,
and moves to $(-1,0)$;
$r_i$ then satisfies $(0,0) > \bm{q}$ in $Z_i$ at time $t+1$ (since $r_i$ is now at $(-1,0)$).
The same number of robots $r_j$ are included in $A_t((-1,0))$,
which should not include a robot $r_j$ satisfying $(0,0) > \bm{q}$ in $Z_j$,
and they all move to $(0,0)$;
$r_j$ then satisfies 
$(0,0) > \bm{q}$ in $Z_j$ at time $t+1$ 
(since $r_i$ is now at $(0,0)$).

Since there is a $gat_1$ robot at $(0,0)$ at $t= 0$,
when the $h$th robot $r_j$ at $(-1,0)$ (which satisfies $(0,0) < \bm{q}$ in $Z_j$)
is activated at some time $t$,
there is no robot at $(0,0)$ which can move to $(-1,0)$.
It is a contradiction, since $P_{t+1}$ is favorable.
\end{proof}

\medskip

A target function $\phi$ is said to be {\em symmetric} (with respect to the origin)
if $\phi(P) = -\phi(-P)$ for all $P \in {\cal P}$.
An algorithm $\Phi$ is said to be {\em symmetric} 
if every target function $\phi \in \Phi$ is symmetric.
Target function $2gat$ is symmetric,
but GATA is not a symmetric algorithm.
Indeed, the next lemma holds.

\begin{lemma}
\label{L1020} 
There is no symmetric algorithm of size 2 for the \mbox{\rm GAT}.
\end{lemma}

\begin{proof}
Proof is by contradiction.
Suppose that $\Phi = \{ \phi_1, \phi_2 \}$ is a symmetric gathering algorithm.

Consider the following case of $n = 4$.
Let ${\cal R} = \{ r_1, r_2, r_3, r_4 \}$,
where $r_1$ and $r_2$, and $r_3$ and $r_4$, are respectively clones.
That is, $r_1$ and $r_2$ (resp. $r_3$ and $r_4$) share 
the same local coordinate system $Z$ (resp. $Z'$),
and $r_1$ and $r_2$ (resp. $r_3$ and $r_4$) take $\phi_1$ (resp. $\phi_2$) 
as their target function.
In the initial configuration $P_0$, $\overline{P_0} = \{ \bm{p}_0, \bm{q}_0\}$,
$\bm{x}_0(r_1) = \bm{x}_0(r_2) = \bm{p}_0$, and $\bm{x}_0(r_3) = \bm{x}_0(r_4) = \bm{q}_0$ hold.
Furthermore, the scheduler activates $r_1$ (resp. $r_3$),
if and only if it activates $r_2$ (resp. $r_4$).

Let ${\cal E}: P_0, P_1, \ldots$ be an execution satisfying the above conditions.
At any time $t$, $\bm{x}_t(r_1) = \bm{x}_t(r_2) = \bm{p}_t$,
and $\bm{x}_t(r_3) = \bm{x}_t(r_4) = \bm{q}_t$, for some $\bm{p}_t$ and $\bm{q}_t$.

We consider the following scheduler:
For any time $t$, unless they all meet at time $t+1$,
the scheduler activates all robots.
Otherwise, it nondeterministically selects either $r_1$ and $r_2$, or $r_3$ and $r_4$,
and activates them.
Since $\cal E$ transforms $P_0$ to a goal configuration, 
in which they all gather, there is a time instant $t$ such that
either $\bm{p}_{t+1} = \bm{p}_t$ and $\bm{q}_{t+1} = \bm{p}_t$,
or $\bm{p}_{t+1} = \bm{q}_t$ and $\bm{q}_{t+1} = \bm{q}_t$.

Now, consider an initial configuration $Q_0$ such that
$\bm{x}_0(r_1) = \bm{x}_0(r_3) = \bm{p}_t$ and $\bm{x}_0(r_2) = \bm{x}_0(r_4) = \bm{q}_t$.
Then under the $\cal FSYNC$ scheduler,
$Q_0 = Q_1 = \ldots$ hold, since $\Phi$ is symmetric.
Thus $\Phi$ is not a gathering algorithm.
\end{proof}

\medskip

There is however a symmetric gathering algorithm 
SGTA $= \{ sgat_1, sgat_2, sgat_3 \}$,
where target functions $sgat_1, sgat_2$, and $sgat_3$
are defined as follows:

\medskip
\noindent
{\bf [Target function $sgat_1$]}
\vspace{-1mm}
\begin{enumerate}
\item
If $m = 1$, then $sgat_1(P) = (0,0)$.

\item
If $m \geq 3$, or $m = 2$ and $k_P \not= 2$,
then $sgat_1(P) = 2gat(P)$.

\item
Suppose that $m = 2$ and $k_P = 2$.
Let $\bm{q}$ be a point in $\overline{P}$ such that $\bm{q}\neq (0,0)$.
Since $(0,0) \in \overline{P}$,
$\bm{q}$ is uniquely determined.
Then $sgat_2(P) = - \bm{q}$.
\end{enumerate}

\medskip
\noindent
{\bf [Target function $sgat_2$]}
\vspace{-1mm}
\begin{enumerate}
\item
If $m = 1$, then $sgat_2(P) = (0,0)$.

\item
If $m \geq 3$, or $m = 2$ and $k_P \not= 2$,
then $sgat_2(P) = 2gat(P)$.

\item
Suppose that $m = 2$ and $k_P = 2$.
Let $\bm{q}$ be a point in $\overline{P}$ such that $\bm{q}\neq (0,0)$.
Since $(0,0) \in \overline{P}$,
$\bm{q}$ is uniquely determined.
Then $sgat_2(P) = - 2 \bm{q}$.
\end{enumerate}

\medskip
\noindent
{\bf [Target function $sgat_3$]}
\vspace{-1mm}
\begin{enumerate}
\item
If $m = 1$, 
then $sgat_3(P) = (0,0)$.

\item
If $m \geq 3$, or $m = 2$ and $k_P \not= 2$,
then $sgat_3(P) = 2gat(P)$.

\item
Suppose that $m = 2$ and $k_P = 2$.
Let $\bm{q}$ be a point in $\overline{P}$ such that $\bm{q}\neq (0,0)$.
Since $(0,0) \in \overline{P}$,
$\bm{q}$ is uniquely determined.
Then $sgat_3(P) = -3 \bm{q}$.
\end{enumerate}

\begin{theorem}
\label{T1020}
\mbox{\rm SGTA} solves the \mbox{\rm GAT} for $n\geq 3$.
The MAS of symmetric algorithm for the \mbox{\rm GAT} is hence $3$.
\end{theorem}

\begin{proof}
By Lemma~\ref{L1020},
it suffices to show that SGTA is a symmetric gathering algorithm. 

It is easy to observe that SGTA is indeed symmetric,
since $2gat$ is symmetric.
By the proof of Theorem~\ref{T1000},
it suffices to show that any execution ${\cal E}: P_0, P_1, \ldots$ 
starting from any unfavorable configuration
eventually reaches a favorable configuration $P_t$.

Without loss of generality,
let us assume that $\overline{P_0} = \{ (0,0), (1,0) \}$,
$\mu_0((0,0)) = \mu_0((1,0)) = h$ for some $h \geq 2$, 
and then $P_t$ is unfavorable for any $t \geq 0$.
Let $R(\bm{p})$ be the set of robots at point $\bm{p}$ at $t = 0$.
Without loss of generality,
suppose that a robot in $R((0,0))$ is activated at $t = 0$.
Then all robots in $R(0,0)$ are activated simultaneously at $t = 0$,
and they all take the same target function, say $sgat_1$,
since, otherwise, $|\overline{P_1}| \geq 3$ holds,
and $P_1$ becomes favorable.
Now $P_1$ is still unfavorable.

By the fairness of the scheduler,
there is a time $t$ such that a robot in $R((1,0))$ is activated.
Then all robots in $R((1,0))$ are activated simultaneously at $t$,
again, because otherwise, $|\overline{P_{t+1}}| \geq 3$ holds.
Observe that there are robots taking $sgat_2$ and $sgat_3$ in $R((1,0))$,
and hence $P_{t+1}$ is favorable since $|\overline{P_{t+1}}| \geq 3$.

\end{proof}

\section{Pattern formation problem}
\label{Spatternformation}

Given a goal pattern $G \in (R^2)^n$ in $Z_0$,
the {\em pattern formation problem} (PF) for $G$ is the problem of 
transforming any initial configuration $I$ into a configuration similar to $G$.
The GAT is the PF for a goal pattern
$G = \{ (0,0), (0,0), \ldots , (0,0) \} \in (R^2)^n$,
and the SCT is reducible to the PF for a right $n$-gon.

\begin{theorem}[\cite{YS10}]
\label{T3010} 
The \mbox{\rm PF} for a goal pattern $G$ is solvable 
(by an algorithm of size $1$) from an initial configuration 
$I$ such that $|I| = |\overline{I}|$, 
if and only if $\sigma(G)$ is divisible by $\sigma(I)$.
The only exception is the \mbox{\rm GAT} for two robots.
\end{theorem}

Thus a pattern $G$ is not formable from a configuration $I$ by an algorithm of size 1,
if $\sigma(G)$ is not divisible by $\sigma(I)$.
In the following, we investigate an algorithm that solves
the PF from any initial configuration $I$, for any $G$.

\begin{lemma}
\label{L3010} 
The MAS for the \mbox{\rm PF} is at least $n$.
\end{lemma}

\begin{proof}
Let $G$ be a right $n$-gon.
If there was a pattern formation algorithm for $G$ with size less $m < n$,
the SCT could be solved by an algorithm of size $m < n$,
which contradicts to Theorem~\ref{T0010}.
\end{proof}

\medskip

A scattering algorithm $n$SCTA transforms any initial configuration $I$
into a configuration $P$ satisfying $P = \overline{P}$.
We can modify $n$SCTA so that the resulting algorithm $n$SCTA$^*$
can transform any initial configuration $I$ into a configuration $P$ 
satisfying ($P = \overline{P}$ and) $\sigma(P) = 1$.
On the other hand, 
there is a pattern formation algorithm (of size 1) for $G$ 
which transforms any initial configuration $P$ satisfying 
($P = \overline{P}$ and) $\sigma(P) = 1$
into a configuration similar to a goal pattern $G$ (see e.g., \cite{YS10}).
The pattern formation problem is thus solvable 
by executing $n$SCTA$^*$ as the first phase, 
and then such a pattern formation algorithm as the second phase,
if we can modify these algorithms so that the robots can 
consistently recognize which phase they are working.
We describe an algorithm PFA = $\{ pf_1, pf_2, \ldots , pf_n \}$ 
(for a swarm of $n$ robots) which solves the PF for a given goal pattern $G$, 
based on this approach.
Since the cases of $n \leq 3$ are trivial,
we assume $n \geq 4$.

We say a configuration $P$ is {\em good},
if $P$ satisfies either one of the following conditions (1) and (2):

\begin{description}
\item[(1)]
$P = \overline{P}$, i.e., $P$ is a set,
and can be partitioned into two subsets $P_1$ and $P_2$
satisfying all of the following conditions:
\begin{description}
\item[(1a)] 
$P_1 = \{ \bm{p}_1 \}$ for some $\bm{p}_1 \in P$.
\item[(1b)]
$dist(\bm{p}_1, \bm{o}_2) \geq 10 \delta_2$,
where $\bm{o}_2$ and $\delta_2$ are respectively the center and the radius of
the smallest enclosing circle $C_2$ of $P \setminus \{ \bm{p}_1 \}$.
\end{description}

\item[(2)]
The smallest enclosing circle $C$ of $P$ contains exactly two points 
$\bm{p}_1, \bm{p}_3 \in P$, 
i.e., 
$\overline{\bm{p}_1 \bm{p}_3}$ 
forms a diameter of $C$.
Consider a (right-handed) $x$-$y$ coordinate system $Z$ satisfying
$\bm{p}_1 = (0,0)$ and $\bm{p}_3 = (31,0)$.\footnote{
Note that $Z$ is uniquely determined,
and the unit distance of $Z$ is $dist(\bm{p}_1,\bm{p}_3)/31$.}
For $i = 1,2,3$, 
let $C_i$ be the unit circle with center $\bm{o}_i$ (and radius 1 in $Z$),
where $\bm{o}_1 = (0,0)$, $\bm{o}_2 = (10,0)$, and $\bm{o}_3 = (30,0)$.
Let $P_i \subseteq P$ 
be the multiset of points included in $C_i$ 
for $i = 1,2,3$.
Then $P$ is partitioned into three submultisets $P_1, P_2$, and $P_3$,
i.e., 
$P \setminus (P_1 \cup P_2 \cup P_3) = \emptyset$,
and $P_1$, $P_2$, and $P_3$ satisfy the following conditions:
\begin{description}
\item[(2a)]
$P_1 = \{ \bm{p}_1 \}$.
\item[(2b)]
$P_2$ is a set (not a multiset).
\item[(2c)]
$P_3$ is a multiset that includes 
$\bm{p}_3$ as a member.
It has a supermultiset $P^*$ which is similar to $G$,
and is contained in $C_3$, 
i.e., $P_3$ is similar to a submultiset $H \subseteq G$.
\end{description}

\end{description}

\begin{lemma}
\label{L3015}
Let $P$ be a good configuration.
Then $P$ satisfies exactly one of conditions (1) and (2),
and $\bm{p}_1$ is uniquely determined in each case.
\end{lemma}

\begin{proof}
Consider any configuration $P$ that satisfies condition (1) for a partition $\{ P_1, P_2 \}$.
Suppose that $P$ also satisfies condition (2).
Then there is a partition $\{ Q_1, Q_2, Q_3 \}$ of $P$ such that
$Q_1$, $Q_2$, and $Q_3$ satisfy conditions (2a), (2b), and (2c), respectively.

Observe that $P_1 = Q_1$ and hence $P_2 = Q_2 \cup Q_3$.
Let $P_1 = \{ \bm{p}_1 \}$ and $Q_1 = \{ \bm{q}_1 \}$.
Suppose $\bm{p}_1 \not= \bm{q}_1$ to derive a contradiction.
Then $\bm{q}_1 \in P_2$,
which implies that every point in $P_2$ is within distance 
$2 \delta_2$ from $\bm{q}_1$.
Since $dist(\bm{p}_1,\bm{q}_1) \geq 9 \delta_2$,
$Q_3 = \{ \bm{p}_1 \}$, and hence $Q_2 = P \setminus \{ \bm{p}_1, \bm{q}_1 \}$.

Let $\bm{o}$ be the center of the smallest enclosing circle of $Q_2$.
Obviously $dist(\bm{q}_1, \bm{o}) \leq 2 \delta_2$.
In $Z$, $\bm{q}_1 = (0,0)$, $\bm{p}_1 = (31,0)$, and $\bm{o} = (10,0)$.
It is a contradiction; 
$dist(\bm{q}_1, \bm{o}) > 2 \delta_2$ holds,
since $dist(\bm{p}_1,\bm{q}_1) \geq 9 \delta_2$.

Thus $P_2$ can be partitioned into two multisets $Q_2$ and $Q_3$, 
and partition $\{ P_1, Q_2, Q_3 \}$ satisfies condition (2).
However, it is impossible by the following reason.
Consider again the center $\bm{o}$ of the smallest enclosing circle of $Q_2$.
By condition (1b), 
$dist(\bm{p}_1, \bm{o}) \geq 9 \delta_2$.
Since $P_2 = Q_2 \cup Q_3$,
$dist(\bm{p}_3, \bm{o}) \leq 2 \delta_2$, on the other hand.
It is a contradiction;
since $\bm{o}$ cannot be placed at $(10,0)$ in $Z$.

By a similar argument, 
if $P$ satisfies condition (2), then it does not satisfy condition (1).

It is also easy to show that $\bm{p}_1$ is unique by a similar argument.
\end{proof}

\medskip

We start with defining $n$SCTA$^*$ = $\{ sct^*_i: i = 1, 2, \ldots , n \}$,
which is a slight modification of $n$SCTA.

\medskip
\noindent
{\bf [Target function $sct^*_i$]}

\medskip
\noindent
(I) If $P$ is good: $sct^*_i(P) = (0,0)$ for $i = 1,2, \ldots , n$.

\medskip
\noindent
(II) If $P$ is not good:
\begin{enumerate}
\item 
For $i = 2, 3, \ldots, n$:

If $P \not=\overline{P}$, then $sct^*_i(P) = sct_i(P)$.

Else if $P =\overline{P}$, then $sct^*_i(P) = (0,0)$.
\item
For $i = 1$: 
\begin{description}
\item[(a)] 
If $P \not=\overline{P}$, then $sct^*_1(P) = sct_i(P)$.
\item[(b)]
If $P = \overline{P}$ and 
$dist((0,0), \bm{o}) < 10 \delta$,
then $sct^*_1(P) = \bm{p}$.
Here $\bm{o}$ and $\delta$ are, respectively, the center and the radius
of the smallest enclosing circle of $P \setminus \{ (0,0) \}$. 
If $\bm{o} \not= (0,0)$,
$\bm{p}$ is the point such that $(0,0) \in \overline{\bm{o}\bm{p}}$
and $dist(\bm{p},\bm{o}) = 10 \delta$.
If $\bm{o} = (0,0)$, $\bm{p} = (10 \delta, 0)$.
\item[(c)]
If $P = \overline{P}$ and $dist((0,0), \bm{o}) \geq 10 \delta$,
then $sct^*_1(P) = (0,0)$.
\end{description}
\end{enumerate}

\begin{lemma}
\label{L3020} 
{\rm $n$SCTA}$^*$ is an algorithm to transform any initial configuration $P_0$
to a good configuration $P$.
\end{lemma}

\begin{proof}
Consider any execution ${\cal E}: P_0, P_1, \ldots$ of $n$SCTA$^*$ from 
any initial configuration $P_0$, which is not good.
By the proof (II) of Theorem~\ref{T0010},
at some time $t$, $P_t = \overline{P_t}$ holds.
If $P_t$ is good, we have nothing to show.

Suppose that $P_t$ is not good.
Let $r_1$ be the robot taking $sct^*_1$ as the target function.
The only robot that can move is $r_1$,
and it computes $sct^*_1(P)$, where $P = \gamma^{-1}_1(P_t)$.
By the definition of $sct^*_1$,
$(P \setminus \{(0,0)\}) \cup \{ \bm{p} \}$ is good.

By definition, 
$\gamma_1((0,0)) = \bm{x}_t(r_1)$,
$P$ and $P_t$ are similar, 
and no other robots move at time $t$.
Thus $P_{t+1} = (P_t \setminus \{ \bm{x}_t(r_1) \}) \cup \{ \gamma_i(\bm{p}) \}$,
which is similar to $(P \setminus \{(0,0)\}) \cup \{ \bm{p} \}$,
although $\gamma_i(\bm{p})$ depends on $Z_1$ when $\bm{o}_2 = (0,0)$ in $Z_1$.
Thus $P_{t+1}$ is good.
\end{proof}

\medskip

Let us explain next how to construct a goal configuration similar to $G$
from a good configuration $P$.

\medskip
\noindent
(I) Suppose that $P$ satisfies condition (1) for a partition $\{P_1, P_2\}$,
where $P_1 = \{ \bm{p}_1 \}$.
If there is a point $\bm{q}$ such that $P_2 \cup \{ \bm{q} \}$ is similar to $G$,
then we move the robot at $\bm{p}_1$ to $\bm{q}$ to complete the formation.

Otherwise, 
let $\bm{p}_3$ be the point satisfying $\bm{o}_2 \in \overline{\bm{p}_1 \bm{p}_3}$
and $dist(\bm{o}_2, \bm{p}_3) = 21 \delta_2$,
where $\bm{o}_2$ and $\delta_2$ are, respectively, the center and the radius of
the smallest enclosing circle $C_2$ of $P_2$.
We choose a point $\bm{p}$ in $P_2$,
and move the robot at $\bm{p}$ to $\bm{p}_3$.
Note that the robot at $\bm{p}$ is uniquely determined,
since $P_2 = \overline{P_2}$.

Then $P$ is transformed into a configuration $P'$ which is good,
and satisfies condition (2) for partition 
$\{ P_1, P_2 \setminus \{ \bm{p} \}, \{ \bm{p}_3 \} \}$.

\medskip
\noindent
(II) Suppose that $P$ satisfies condition (2) for a partition $\{P_1, P_2, P_3\}$,
where $P_1 = \{ \bm{p}_1 \}$.
Like the above case,
we choose a point $\bm{p}$ in $P_2$, 
and move the robot 
at $\bm{p}$ to a point $\bm{q}$.
Here $\bm{q}$ must satisfy that
there is a superset $P^*$ of $P_3 \cup \{ \bm{q} \}$ 
which is contained in $C_3$, and is similar to $G$.

By repeating this transformation, 
a configuration $P$ satisfying condition (1) is
eventually reached, when $P_2$ becomes empty,
$P_3$ now is similar to a submultiset $H$ of $G$.

To carry out this process,
we need to specify (i) $\bm{p} \in P_2$ in such a way that all robots 
can consistently recognize it,
and (ii) $\bm{p}_3$ in (I) and $\bm{q}$ in (II).

We define a point $\bm{p} \in P_2$.
When $|P_2| = 1$, $\bm{p}$ is the unique element of $P_2$.
When $|P_2| \geq 2$, let $P_{12} = P_1 \cup P_2$.
Then $k_{P_{12}} = 1$ by the definition of $\bm{p}_1$.
Since $k_{P_{12}} = 1$, $\succ_{P_{12}}$ is a total order on $P_{12}$
(and hence on $P_2$), which all robots in $P$ (in particular, in $P_2$) can compute.
Let $\bm{p} \in P_2$ be the largest point in $P_2$ with respect to $\succ_{P_{12}}$.
Since $P_2$ is a set, the robot $r$ at $\bm{p}$ is uniquely determined,
and $r$ (or its target function) knows that it is the robot to move to $\bm{p}_3$ or $\bm{q}$.

We define the target points $\bm{p}_3$ and $\bm{q}$.
It is worth emphasizing that $r$ can choose the target point by itself, 
and the point is not necessary to share by all robots.

Point $\bm{p}_3$ is uniquely determined.
To determine $\bm{q}$,
note that $P_3$ has a supermultiset $P^*$ which is similar to $G$,
and is contained in $C_3$.
Thus $r$ arbitrarily chooses such a multiset $P^*$,
and takes any point in $P^* \setminus P_3$ as $\bm{q}$.
(There may be many candidates for $P^*$.
Robot $r$ can choose any one, 
e.g., the smallest one in terms of $\sqsupset$ in its $x$-$y$ local coordinate system.)

Using points $\bm{p}$, $\bm{p}_3$, and $\bm{q}$ defined above,
we describe a pattern formation algorithm PFA = $\{ pf_1, pf_2, \ldots , pf_n \}$
for a goal pattern $G$,
where target functions $pf_i (i = 1, 2, \ldots, n)$ are defined as follows:

\medskip
\noindent
{\bf [Target function $pf_i$]}

\begin{enumerate}
 \item 
When $P$ is not good:
$pf_i(P) = sct^*_i(P)$.

\item
When $P$ is a good configuration satisfying condition (1):
\begin{description}
\item[(2a)] 
Suppose that there is a $\bm{q}$ such that $P_2 \cup \{ \bm{q} \}$ is similar to $G$.
Then $pf_i(P) = \bm{q}$ if $(0,0) \in P_1$;
otherwise, $pf_i(P) = (0,0)$.

\item[(2b)]
Suppose that there is no point $\bm{q}$ such that $P_2 \cup \{ \bm{q} \}$ is similar to $G$.
Then $pf_i(P) = \bm{p}_3$ if $(0,0)$ is the largest point in $P_2$ 
with respect to 
$\succ_{P_{12}}$;
otherwise, $pf_i(P) = (0,0)$.
\end{description}

\item
When $P$ is a good configuration satisfying condition (2):
$pf_i(P) = \bm{q}$ if $(0,0)$ is the largest point in $P_2$
with respect to $\succ_{P_{12}}$;
otherwise, $pf_i(P) = (0,0)$.
\end{enumerate}

\begin{theorem}
\label{T3020}
The MAS for the \mbox{\rm PF} is $n$.
\end{theorem}

\begin{proof}
By Lemma~\ref{L3010},
the MAS for the \mbox{\rm PF} is at least $n$.

On the other hand,
PFA is an algorithm of size $n$ for the PF,
by the arguments above and the description of PFA.
\end{proof}

\section{Fault tolerant scattering problems}
\label{Sfsct}

A fault means a crash fault in this paper.
The {\em f-fault tolerant $c$-scattering problem} ($f$F$c$S) is
the problem of transforming any initial configuration to
a configuration $P$ such that $|\overline{P}| \geq c$,
as long as at most $f$ robots have crashed.

\begin{observation}
\label{Ofs0010} 
\begin{enumerate}
 \item 
\mbox{\rm 1SCTA} solves the \mbox{\rm $f$F1S} for all $1 \leq f \leq n$,
since $|\overline{P}| \geq 1$ for any configuration $P$.
The MAS for the \mbox{\rm $f$F1S} is thus $1$ for all $1 \leq f \leq n$.

\item
The MAS for the \mbox{\rm $n$F$c$S} is $\infty$ for all $2 \leq c \leq n$,
since $|\overline{P_0}| = |\overline{P_t}| = 1$ holds for all $t \geq 0$,
if $|\overline{P_0}| = 1$, and all robots have crashed at time $0$.
\end{enumerate}
\end{observation}

\begin{lemma}
\label{Lfs0010} 
For all $1 \leq f \leq n-1$ and $2 \leq c \leq n$,
the \mbox{$f$F$c$S} is unsolvable, 
if $c+f-1 > n$.
\end{lemma}

\begin{proof}
Suppose that $f (> 0)$ faulty robots have crashed at the same point in $P_0$.
For any configuration $P$ reachable from $P_0$,
$|\overline{P}| \leq n-f+1$.
Thus, if $c > n-f+1$, then the $f$F$c$S is unsolvable.
\end{proof}

\begin{lemma}
\label{Lfs0015}
\begin{enumerate}
 \item 
The \mbox{\rm $f$F2S} is unsolvable if $f \geq n-1$.
\item
If $1 \leq f \leq n-2$,
\mbox{\rm $(f+2)$SCTA} solves the \mbox{\rm $f$F2S}.
\item
The MAS for the \mbox{\rm $f$F2S} is $\infty$ if $f = n-1$;
otherwise, it is $f+2$.
\end{enumerate}
\end{lemma}

\begin{proof}
(I) We first show that the MAS for the $f$F2S is at least $f+2$.
Proof is by contradiction.
Suppose that there is an algorithm $\Phi = \{ \phi_1, \phi_2, \ldots , \phi_m \}$
of size $m < f+2$ for the $f$F2S, to derive a contradiction.

Let $\overline{P_0} = \{ (0,0) \}$.
We first show that there is $i$ such that $\phi_i((0,0)) = (0,0)$.
To derive a contradiction,
suppose that 
$\phi_i((0,0)) = \bm{p}_i \not= (0,0)$ 
for $i = 1, 2, \ldots , m$.
Consider any robot $r_j$ and assume that its target function is $\phi_i$.
We choose the $x$-$y$ local coordinate system $Z_j$ of $r_j$ 
to satisfy $\gamma_j(\bm{p}_i) = (1,0)$.
Then under the $\cal FSYNC$ scheduler,
$|\overline{P_t}| = 1$ for all $t \geq 0$.
It is a contradiction.

Without loss of generality, we assume that $\phi_m((0,0)) = (0,0)$.
Consider the following assignment:
For all $i = 1, 2, \ldots , m-1$,
robot $r_i$ takes target function $\phi_i$,
and for all $j = m, m+1, \ldots , n$,
robot $r_j$ takes target function $\phi_m$.
We also assume that all robots $r_i (i = 1, 2, \ldots , m-1)$
have crashed at time $0$.
Note that $m-1 \leq f$, because $m < f+2$.

Then under the $\cal FSYNC$ scheduler,
$P_t = P_0$ for all $t \geq 0$.
It is a contradiction.
Thus the MAS for the $f$F2S is at least $f+2$.

Since the MAS for the $f$F2S is less than or equal to $n$,
the $(n-1)$F2S is unsolvable (since the MAS is at least $n+1$).

\medskip
\noindent
(II) We next show that $(f+2)$SCTA solves the $f$F2S, if $f \leq n-2$.

Proof is by contradiction.
Consider any initial configuration $P_0$ satisfying $|\overline{P_0}| = 1$.
We assume that $|\overline{P_t}| = 1$ for all $t \geq 0$,
to derive a contradiction.
There is a non-faulty robot $r$ that does not take $sct_1$,
and there is another robot $r'$ that takes $sct_1$ by definition.

Since the $\cal SSYNC$ scheduler is fair,
there is a time instant at which $r$ is activated for the first time.
Let $\overline{P_t} = \{ \bm{p} \}$.
Then $\bm{x}_{t+1}(r') = \bm{p}$ (regardless of whether it is faulty or not).
On the other hand, $\bm{x}_{t+1}(r) \not= \bm{p}$,
since $sct_i((0,0)) = (1,0)$ for all $i \not= 1$.
Thus $|\overline{P_{t+1}}| \geq 2$.
It is a contradiction.

Thus the MAS for the \mbox{\rm $f$F2S} is $f+2$, if $f \leq n-2$.
\end{proof}

\begin{lemma}
\label{Lfs0020}
Suppose that $1 \leq f \leq n-1$, $3 \leq c \leq n$, and $c+f-1 \leq n$.
Algorithm \mbox{\rm $(c+f-1)$SCTA} solves the \mbox{\rm $f$F$c$S}.
\end{lemma}

\begin{proof}
First we show that $(c+f-1)$SCTA solves the $f$F$c$S,
from any initial configuration $P_0$ satisfying $|\overline{P_0}| \geq 2$.

Let $P_t$ be a configuration at time $t$.
Suppose that $|\overline{P_t}| \geq 2$,
and a non-faulty robot $r_i$ is activated at $t$.
By the proof of Theorem~\ref{T0010},
if $r_i$ and $r_j$ have different target functions,
then $\bm{x}_{t+1}(r_i) \not= \bm{x}_{t+1}(r_j)$,
independently of their current positions $\bm{x}_t(r_i)$ and $\bm{x}_t(r_j)$,
and regardless of whether or not $r_j$ is activated.
Thus $|\overline{P_t}| \leq |\overline{P_{t+1}}|$.
To complete the proof,
we claim: If $|\overline{P_t}| < c$,
then $|\overline{P_t}| < |\overline{P_{t'}}|$ for some $t' > t$.

For any point $\bm{p} \in \overline{P_t}$,
let $R(\bm{p})$ be the set of robots which reside at $\bm{p}$.
If $R(\bm{p})$ includes two robots such that 
(1) both of them are non-faulty and have different target functions,
or (2) exactly one of them is faulty, then the claim holds.

Since $(c+f-1) - f = c-1$,
there are at least $c-1$ non-faulty robots whose target functions are distinct.
By the argument above,
there is no point $\bm{p} \in \overline{P_t}$
such that $R(\bm{p})$ includes two non-faulty robots 
having different target functions.
Since $|\overline{P_t}| < c$, 
$|\overline{P_t}| = c-1$,
and there is a non-faulty robot in each $R(\bm{p})$.
It is because, otherwise,
two non-faulty robots having different target functions reside at a point.
Since $f \geq 1$, there is a point $\bm{p}$ such that $R(\bm{p})$
includes both non-faulty and faulty robots,
and the claim hold.

Next we show that $(c+f-1)$SCTA solves the $f$F$c$S,
from any initial configuration $P_0$.
By Lemma~\ref{Lfs0015},
$(f+2)$SCTA solves the $f$F2S, if $f+2 \leq n$.
Since $c \geq 3$, 
$f+2 \leq c+f-1 \leq n$,
which implies that $(c+f-1)$SCTA solves the $f$F2S
by a similar argument to the proof (II) of Lemma~\ref{Lfs0015}.
Thus, $(c+f-1)$SCTA solves the $f$F$c$S from any initial configuration.
\end{proof}

\begin{theorem}
\label{Tfs0010} 
Suppose that $1 \leq f \leq n-1$ and $2 \leq c \leq n$.
\begin{enumerate}
 \item 
If $c =2$, the MAS for the \mbox{\rm $f$F2S} is $\infty$ if $f = n-1$;
otherwise, if $1 \leq f \leq n-2$, the MAS for the \mbox{\rm $f$F2S} is $f+2$.
Indeed, \mbox{\rm $(f+2)$SCTA} solves the \mbox{\rm $f$F2S}, if $1 \leq f \leq n-2$.

\item
If $3 \leq c \leq n$, the MAS for the \mbox{\rm $f$FcS}
is $\infty$ if $c+f-1 > n$;
otherwise, if $c+f-1 \leq n$, the MAS for the \mbox{\rm $f$F$c$S} is $c+f-1$.
Indeed, \mbox{\rm $(c+f-1)$SCTA} solves the \mbox{\rm $f$F$c$S},
if $c+f-1 \leq n$.
\end{enumerate}
\end{theorem}

\begin{proof}
The case of $c=2$ holds by Lemma~\ref{Lfs0015}.

As for the case of $3 \leq c \leq n$,
by Lemmas~\ref{Lfs0010}  and \ref{Lfs0020},
it suffices to show that the $f$F$c$S is not solvable
by an algorithm of size less than $c+f-1$,
provided that $c+f-1 \leq n$ and $c \geq 3$.

Suppose that there is an algorithm of size $m < c+f-1$ for the $f$F$c$S, 
to derive a contradiction.
Without loss of generality, we may assume that $m = c+f-2 \geq f+1$.

The proof is almost the same as the proof (I) of Theorem~\ref{T0010}.
Let $\Phi = \{ \phi_1, \phi_2, \ldots , \phi_m \}$ be any algorithm 
for the $f$F$c$S.
Consider the following situation:
\begin{enumerate}
\item 
All robots $r_i$ ($1 \leq i \leq n$) share the unit length 
and the direction of positive $x$-axis.
\item
A target function assignment $\cal A$ is defined as follows:
${\cal A}(r_i) = \phi_i$ for $1 \leq i \leq m-1$,
and ${\cal A}(r_i) = \phi_m$ for $m \leq i \leq n$.
\item
All robots initially occupy the same location $(0,0)$,
i.e.,  $\overline{P_0} = \{ (0,0) \}$.
\item
All robots $r_i$ ($1 \leq i \leq f$) have crashed at time $0$.
\item
The scheduler is ${\cal FSYNC}$.
\end{enumerate}

Let ${\cal E}: P_0, P_1, \ldots$ be the execution of $\cal R$ starting from $P_0$,
under the above situation.
By an easy induction on $t$,
all faulty robots $r_i$ ($1 \leq i \leq f$) occupy $(0,0)$, for all $t \geq 0$,
and all robots $r_i$ ($f+1 \leq i \leq n$) with the same target function occupy 
the same location, for all $t \geq 0$. 
Thus, $|\overline{P_t}| \leq m-f+1 = (c+f-2)-f+1 = c-1 < c$,
which implies that $\Phi$ does not solve the $f$F$c$S, for $c\geq 3$.
\end{proof}

\section{Fault tolerant gathering problems}
\label{Sfgp}

\subsection{At most one robot crashes}
\label{SS1fg}

The {\em f-fault tolerant gathering problem} ($f$FG)
is the problem of gathering {\bf all non-faulty robots} at a point,
as long as at most $f$ robots have crashed.
The {\em f-fault tolerant gathering problem to $f$ points} ($f$FGP)
is the problem of gathering {\bf all robots} (including faulty ones)
at $f$ (or less) points, as long as at most $f$ robots have crashed.
We omit $f$ from the abbreviations to write FG and FGP when $f = 1$.

\begin{theorem}
\label{T4010} 
Provided that at most one robot crashes,
\mbox{\rm 2GATA} transforms any initial configuration 
into a configuration $P$ satisfying one of the following conditions:
\begin{enumerate}
 \item 
$|\overline{P}| = 1$,
\item
$n$ is even,
$\overline{P} = \{ \bm{p}_1, \bm{p}_2 \}$,
and $\mu_P(\bm{p}_1) = \mu_P(\bm{p}_2) = n/2$, or
\item
$\overline{P} = \{ \bm{p}_1, \bm{p}_2 \}$,
$\mu_P(\bm{p}_1) = n-1$, $\mu_P(\bm{p}_2) = 1$,
and the robot at $\bm{p}_2$ has crashed.
\end{enumerate}
\end{theorem}

\begin{proof}
We associate a triple $\lambda(P) = (k_P, m_P, -\mu_P)$
with a configuration $P$ of $n$ robots,
where $m_P = |\overline{P}|$,
and $\mu_P$ is 
$\mu_P(\bm{o}_P)$ if $k_P \geq 2$;
otherwise, it is $\mu_P(\bm{p})$,
where $\bm{p}$ is the largest point in $\overline{P}$ with respect to $\succ_P$.
Let $<$ be the lexicographic order on {\bf Z}$^3$,
i.e., $(a,b,c) < (a',b',c')$, 
if and only if (1) $a < a'$, (2) $a = a'$ and $b < b'$, 
or (3) $a = a'$, $b = b',$ and $c < c'$ holds,
where {\bf Z} is the set of integers.

For any given initial configuration $P_0$,
let $V$ be the set of configurations reachable from $P_0$,
i.e., those which occur in an execution starting from $P_0$.
We draw a directed graph $DG = (V, \vdash)$,
which represents the transition relation defined by 2GATA starting from $P_0$,
i.e., for any two distinct configurations $P,Q \in V$,
$P \vdash Q$, if and only if $Q$ is directly reached from $P$ 
by activating some robots.

By the proof of Theorem~\ref{T1000},
if $P \vdash Q$, then $\lambda(P) > \lambda(Q)$,
and hence $DG$ is a directed acyclic graph with sinks $P_f$, 
where $\lambda(P_f)$ is either $(0,1,-n)$ or $(2,2,0)$.

Suppose that for any configuration $P$ 
there are (at least) two distinct configurations $Q$ and $Q'$ such that 
$P \vdash Q$ and $P \vdash Q'$.
Let $M(P,P')$ be the set of robots 
which move during the transition $P \vdash P'$.
Since $Q \not= Q'$, $M(P,Q) \not= M(P,Q')$.
We claim that even if a robot, say $r$, has crashed, 
there remains a transition $P \vdash Q''$ in $DG$.
That is, we show that, for any robot $r$,
there is a $Q''$ such that $P \vdash Q''$ and $r \not\in M(P,Q'')$.

If $M(P,Q) \not= \{ r \}$, then there is a robot $r' (\not= r) \in M(P,Q)$.
Let $Q''$ be the configuration reachable from $P$ by activating $r'$ alone.
Then $P \vdash Q''$ and $r \not\in M(P,Q'')$.
Otherwise, if $M(P,Q) = \{ r \}$, since $M(P,Q) \not= M(P,Q')$,
there is a robot $r' (\not= r) \in M(P,Q')$.
By the same reason as above, 
there is a $Q''$ such that $P \vdash Q''$ and $r \not\in M(P,Q'')$.

Consider an acyclic directed graph $DG'$ constructed from $DG$
by removing all transitions $P \vdash Q$ such that $r \in M(P,Q)$.
Since we assume that every configuration $P$ has at least two transitions in $DG$, 
there is a transition $P \vdash Q'$ in $DG$ 
such that $r \not\in M(P,Q)$ for all $P$, 
i.e., $DG'$ still has sinks $P_f$.

However, there exists a configuration $P$ having exactly one transition 
to another configuration.
Consider a configuration $P$ such that $\lambda(P) = (1,2,-n+1)$.
$P$ has exactly one transition;
if $\overline{P} = \{ \bm{p}_1, \bm{p_2} \}$,
$\mu_P(\bm{p}_1) = n-1$, and $\mu_P(\bm{p}_2) = 1$,
then the unique robot at $\bm{p}_2$ moves to $\bm{p}_1$,
and the other robot do not move even if they are activated.
Recall that a configuration $P$ such that $\lambda(P)$ is 
either $(0, 1, -n)$ or $(2, 2, 0)$ has no transition in $DG$.
Except for these configurations, 
i.e., if $\lambda(P) \not\in \{(1,2,-n+1), (0,1,-n), (2, 2, 0) \}$,
there are two transitions from $P$ in $DG$,
which is obvious from the definition of Steps~2 and 3 of $2gat$.  
Thus, $DG'$ has sinks $P'_f$ such that $\lambda(P'_f)$ is 
$(1, 2, -n+1)$, $(0,1,-n)$, or $(2,2,0)$.

If $P$ such that $\lambda(P) = (1,2,-n+1)$ has no transition in $DG'$,
then $P$ satisfies $\overline{P} = \{ \bm{p}_1, \bm{p}_2 \}$,
$\mu_P(\bm{p}_1) = n-1$, $\mu_P(\bm{p}_2) = 1$,
and the robot at $\bm{p}_2$ has crashed.
\end{proof}

\medskip

We have a corollary to Theorem~\ref{T4010}.

\begin{corollary}
\label{C4020}
If $n$ is odd,
\mbox{\rm 2GATA} is an algorithm for the \mbox{\rm FG}.
Thus the MAS for the \mbox{\rm FG} is 1, if $n$ is odd.
\end{corollary}

2GATA is not a gathering algorithm 
(and hence does not solve the FG),
since it does not change an unfavorable configuration.
The case in which $n$ is even remains.
Indeed, we have the following theorem.

\begin{lemma}
\label{L4010} 
The MAS for the \mbox{\rm FG} is at least $3$.
\end{lemma}

\begin{proof}
Proof is by contradiction.
Suppose that there is a fault tolerant gathering algorithm 
$\Phi = \{ \phi_1, \phi_2 \}$ of size 2.

We follow the proof of Lemma~\ref{L1020}.
Consider the case of $n = 4$.
Let ${\cal R} = \{ r_1, r_2, r_3, r_4 \}$.
Then there is a configuration $P = \{ \bm{p}, \bm{p}, \bm{q}, \bm{q} \}$
such that a $\phi_1$ robot $r$ at $\bm{p}$ does not move for some $x$-$y$ 
local coordinate system $Z$
and a $\phi_2$ robot $r'$ at $\bm{q}$ moves to $\bm{p}$  for some $x$-$y$ 
local coordinate system $Z'$.

We consider the following configuration:
Robots $r_1$ and $r_2$ with $Z$ occupy $\bm{p}$.
Robot $r_3$ occupies $\bm{q}$,
whose $x$-$y$ local coordinate system is a one constructed from $Z$
by rotating $\pi$ about the origin.
Robots $r_1$, $r_2$, and $r_3$ take $\phi_1$ as their target functions.
Robot $r_4$ takes $\phi_2$, occupies $\bm{q}$, and has crashed.
Since all of them do not move,
it is a contradiction.
\end{proof}

\medskip

However, there is a symmetric algorithm of size 3 for the FG.

\begin{theorem}
\label{T4040} 
\mbox{\rm SGTA} = $\{ sgat_1, sgat_2, sgat_3\}$ solves the \mbox{\rm FG} for $n\geq 3$.
The MAS for the FG is $3$.
\end{theorem}

\begin{proof}
It suffices to show that SGTA solves the FG,
since the MAS for the FG is at least 3 by Lemma~\ref{L4010}.
The proof is almost the same as that of Theorem~\ref{T1020}.

By Theorem~\ref{T4010}, 
it suffices to show that any execution ${\cal E}: P_0, P_1, \ldots$ 
starting from any unfavorable configuration
eventually reaches a favorable configuration $P_t$,
in the presence of at most one faulty robot.
We follow the proof of Theorem~\ref{T1020}.
If there is a non-faulty robot which takes $sgat_i$ for $i = 1,2,3$,
the proof of Theorem~\ref{T1020} correctly works as a proof of this theorem.

All what we need to show is the case in which
a target function, say $sgat_1$, is taken by a unique robot $r$, and $r$ crashes,
so that there is no non-faulty robot taking $sgat_1$ after the crash.
However, the same argument holds,
since the ''target point'' of $r$, i.e., the current position, 
is different from the target points of the other robots.
\end{proof}

\medskip

Let us next consider the FGP.
The FGP is definitely not easier than the FG by definition.
However, you might consider that the difference of difficulty 
between the FG and the FGP would be subtle,
since the problem of converging all robots (including faulty ones) at a point
in the presence of at most one faulty robot (the counterpart of the FGP in convergence)
is solvable by a simple Goto-Gravity-Center algorithm (CoG) of size 1.\cite{AY23}

\begin{theorem}
\label{T4050}
The \mbox{\rm FGP} is unsolvable.
That is, the MAS for the \mbox{\rm FGP} is $\infty$.
\end{theorem}

\begin{proof}
To derive a contradiction,
suppose that there is an algorithm $\Phi$ for the FGP.
Consider any execution ${\cal E}: P_0, P_1, \ldots$ with initial configuration
$P_0$ such that $|\overline{P_0}| > 1$ under a central schedule $\cal S$,
which activates exactly one robot at a time.
(Note that any central schedule is an $\cal SSYNC$ schedule.)
Provided that no robots crash,
$\cal S$ tries not to lead $\cal E$ to a goal configuration as long as possible
(i.e., $\cal S$ is an adversary for $\Phi$).

Since $\Phi$ is an algorithm for the FGP,
no matter which $\cal S$ is given,
there is a time $t$ such that $\overline{P_t} = \{ \bm{q}_1, \bm{q}_2 \}$,
$\mu_t(\bm{q}_1) = n-1$, $\mu_t(\bm{q}_2) = 1$,
and the destination point of each robot at $\bm{q}_1$ is $\bm{q}_1$.

If the robot at $\bm{q}_2$ crashes at $t$,
then for all $t' > t$, $P_{t'} = P_t$, a contradiction.
\end{proof}

\subsection{At most $f$ robots crash}
\label{SSffg}

{\em The $f$-fault tolerant $c$-gathering problem} ($f$F$c$G) is the problem of 
gathering {\bf all non-faulty robots} at $c$ (or less) points,
as long as at most $f$ robots have crashed.
{\em The $f$-fault tolerant $c$-gathering problem to $c$ points} ($f$F$c$GP) is 
the problem of gathering {\bf all robots} (including faulty ones) 
at $c$ (or less) points,
as long as at most $f$ robots have crashed.
When $c = 1$, $f$F$c$G is abbreviated as $f$FG,
and $f$F$c$GP is abbreviated as $f$FGP when $c = f$.
The $f$F$c$G is not harder than the $f$F$c$GP by definition.
In general, the $f$F$c$GP is not solvable if $c < f$.
Theorem~\ref{T4050} shows that $f$FGP is unsolvable when $f = 1$.
As a corollary of Theorem~\ref{T4010}, we have:

\begin{corollary}
\label{C4010}
\mbox{\rm 2GATA} solves the \mbox{\rm $1$F$2$GP}.
The MAS for the \mbox{\rm $1$F$2$GP} is $1$.
\end{corollary}

SGTA transforms any goal configuration of the 1F2G
into a goal configuration of the FG by Theorem~\ref{T4040},
but the transformation of any goal configuration of the 1F2G
into a goal configuration of the FGP is unsolvable by the proof of Theorem~\ref{T4050},
which shows the difference between FG and FGP.
For a configuration $P_t$,
let $m_t =  |\overline{P_t}|$, $k_t = k_{P_t}$, $\mu_t = \mu_{P_t}$,
$C_t = C_{P_t}$, $\bm{o}_t = \bm{o}_{P_t}$, $CH_t = CH(P_t)$,
and $\succ_t = \succ_{P_t}$.

\begin{lemma}
\label{L4020} 
For any $1 \leq f \leq n-1$,
\mbox{\rm 2GATA} solves the \mbox{\rm $f$FG} from any initial configuration $P_0$
satisfying $k_0 = 1$.
\end{lemma}

\begin{proof}
Consider any initial configuration $P_0$ satisfying $k_0 = 1$,
and let ${\cal E}: P_0, P_1, \ldots$ be any execution,
provided that at most $f$ robots crash.
Let $F_t$ be the set of faulty robots at time $t$.
For any $t$, $F_t \subseteq F_{t+1}$.
Let $F = \lim_{t \rightarrow \infty} F_t$.
Then $|F| \leq f$.

By the definition of 2GATA,
$k_t = 1$ for all $t$,
and every non-faulty robots activated at $t$ moves to $\bm{p}_0$,
where $\bm{p}_0 \in \overline{P_0}$ is the largest point with respect to $\succ_0$.
Note that $\bm{p}_0 \in \overline{P_t}$ and 
it is still the largest point in $\overline{P_t}$ with respect to $\succ_t$.

If $r \not\in F$ and $r$ is not at $\bm{p}_0$,
then eventually $r$ is activated and move to $\bm{p}_0$.
Thus all non-faulty robots eventually gather at $\bm{p}_0$.
\end{proof}

\medskip

Note that even if a configuration $P_t$ such that
$k_t = 1$ and $\mu_t(\bm{p_0}) = n-1$ is reached,
we cannot terminate 2GATA (to solve the $f$FG),
since the robot not at $\bm{p_0}$ may be non-faulty.

\begin{lemma}
\label{L4030} 
For any $1 \leq f \leq n-1$,
provided that at most $f$ robots crash,
\mbox{\rm 2GATA} transforms any configuration $P_0$ 
into a configuration $P$ satisfying one of the following conditions :
\begin{enumerate}
 \item 
$P$ is a goal configuration of the \mbox{\rm $f$FG}, or
\item
$\overline{P} = \{ \bm{p}_1, \bm{p}_2 \}$,
and $\mu_{P}(\bm{p_1}) = \mu_P(\bm{p}_2) = n/2$.
\end{enumerate}
\end{lemma}

\begin{proof}
Consider any execution ${\cal E}:P_0, P_1, \ldots$ starting from any configuration $P_0$.
Proof is by contradiction.
By Lemma~\ref{L4020}, to derive a contradiction,
assume that $k_t \geq 2$ for all $t \geq 0$.

We follow the proof of Theorem~\ref{T4010},
and construct a directed acyclic graph $DG = (V, \vdash)$ for $P_0$,
where $V$ is the set of configurations reachable from $P_0$ by 2GATA.
It represents the transition relation between configurations,
when no robots are faulty.
However, any execution $\cal E$ in which some robots may crash is 
also represented by a path in $DG$ from $P_0$.

Consider any configuration $P_t$ such that $k_t \geq 2$.
Let $A_t$ be the set of robots not at $\bm{o}_t$,
and let $F_t$ be the set of faulty robots at $t$.
If $A_t \subseteq F_t$,
all non-faulty robots are at $\bm{o}_t$,
and hence $P_t$ is a goal configuration of the $f$FG.

If $A_t \not\subseteq F_t$, 
since there is a non-faulty robot $r$ not at $\bm{o}_t$,
there is a configuration $Q$ such that $P_t \vdash Q$
and $\lambda(P_t) > \lambda(Q)$.

Since $DG$ is a directed acyclic graph,
eventually the $f$FG is solved (satisfying $A_t \subseteq F_t$),
or $P_t$ reaches a sink $P_f$ of $DG$ 
(since $(V,>)$ is a well-founded set).
\end{proof}

\begin{theorem}
\label{T4025} 
\mbox{\rm 2GATA} solves both the \mbox{\rm $f$F2G} and the \mbox{\rm $f$F($f+1$)GP},
for all $f = 1, 2, \ldots , n-1$.
The MAS for each of the problems \mbox{\rm $f$F2G} and \mbox{\rm $f$F($f+1$)GP} is 1,
for all $f = 1, 2, \ldots , n-1$.
\end{theorem}

\begin{proof}
Immediate by Lemmas~\ref{L4020} and \ref{L4030}. 
\end{proof}

\medskip

We already know that SGTA solves $f$FG if $f = 1$ by Theorem~\ref{T4040}.

\begin{theorem}
\label{T4060}
\mbox{\rm SGTA} solves the \mbox{\rm $f$FG}
for all $f = 2, 3, \ldots, n-1$.
The MAS for the \mbox{\rm $f$FG} is $3$.
\end{theorem}

\begin{proof}
Since there is no algorithm of size 2 for the 2FG by Lemma~\ref{L4010},
it suffices to show that SGTA solves $f$FG.
Consider any execution ${\cal E}: P_0, P_1, \ldots$
starting from any configuration $P_0$.
We show that eventually $\cal E$ reaches a goal configuration of $f$FG,
provided that at most $f$ robots have crashed.
By Lemma~\ref{L4030},
we may assume that $P_0$ is unfavorable,
i.e., $m_0 = 2$ and $k_0 = 2$ hold.
By the definition of SGTA,
$P_t$ is linear for all $t \geq 0$.

Suppose first that there is a time $t$ such that $P_t$ is 
(linear and) favorable, 
i.e., either $m_t \not= 2$ or $k_t \not= 2$ holds.
Like the proof of Lemma~\ref{L4030},
consider $DG = (V, \vdash)$ for $P_t$.
Using the same argument as the proof of Lemma~\ref{L4030},
eventually the $f$FG is solved (satisfying $A_t \subseteq F_t$),
or $P_t$ reaches a sink $P_f$ of $DG$ 
(since $(V,>)$ is a well-founded set).
By (I) of the proof of Theorem~\ref{T1000}, 
the sink $P_f$ is a goal configuration of $f$FG,
i.e., $m_{P_f} = 1$.

Without loss of generality,
We may assume that $\overline{P_0} = \{ \bm{p}_1, \bm{p}_2 \}$,
and $\mu_0(\bm{p}_1) = \mu_0(\bm{p}_2) = n/2$.
We assume that $m_{P_t} = 2$ and $k_{P_t} = 2$ for all $t \geq 1$,
to derive a contradiction.
We follow the proof of Theorem~\ref{T4040}.
Let $R_i$ be the robots initially at $\bm{p}_i$ for $i = 1,2$.
Suppose that a non-faulty robot at $\bm{p}_1$, say $r_1$, is activated at time $0$.
Then all robots at $\bm{p}_1$ are non-faulty and must be activated simultaneously.
Furthermore, they all take the same target function, say $sgat_1$,
since otherwise $m_1 \geq 3$ holds.
As long as the robots $R_1$ are activated,
a configuration $P_t$ satisfies $\overline{P_t} = \{ \bm{q}_1, \bm{p}_2 \}$
for some $\bm{q}_1 \in R^2$, and $\mu_0(\bm{q}_1) = \mu_0(\bm{q}_2) = n/2$,
where the set of robots in $R_1$ is at $\bm{q}_1$.

Since the scheduler is fair,
let $t$ be the first time instant that a non-faulty robot in $R_2$, 
say $r_2$, is activated.
Then all robots at $\bm{p}_2$ are non-faulty and must be activated simultaneously.
Furthermore, they all take the same target function.
It is a contradiction, 
since both $sgat_2$ and $sgat_3$ are taken by some robots in $R_2$.
Thus there is a time $t$ such that $P_t$ is linear and favorable,
and the proof completes.
\end{proof}

\begin{theorem}
\label{T4070} 
The \mbox{\rm $f$FGP} is unsolvable for $f = 2, \ldots , n-1$.
\end{theorem}

\begin{proof}
A configuration $P$ is a goal configuration if $m_P = |\overline{P}| \leq f$.
Proof is by contradiction.
Suppose that there is an algorithm $\Phi$ for the $f$FGP.

We arbitrarily choose a configuration $P_0$ such that $m_0 > f$,
and consider any execution ${\cal E}: P_0, P_1, \ldots$ from $P_0$,
provided that no crashes occur,
under a schedule $\cal S$ we specify as follows:
For $P_t$, let $A_t$ be a largest set of robots such that 
its activation does not yield a goal configuration.
If there are two or more such largest sets, 
$A_t$ is an arbitrary one.
Then $\cal S$ activates all robots in $A_t$ at time $t$,
and the execution reaches $P_{t+1}$, 
which is not a goal configuration.
($A_t$ may be an empty set.)

Since $\cal E$ does not reach a goal configuration forever,
$\cal S$ violates the fairness,
because $\Phi$ is an algorithm for $f$FGP.
Let $U$ be the set of robots which are not activated 
infinitely many times in $\cal E$.
We assume that $\cal S$ makes $U$ as small as possible.

We show that $|U| \leq f$.
The proof is by contradiction.
Suppose that $f < |U|$ to derive a contradiction.
Let $t_0$ be a time such that any robot in $U$ is not activated thereafter.
The positions of the robots in $U$ at $t_0$ is denoted by 
$Q = \{\bm{x}_{t_0}(r) \mid r \in U\}$.
Then obviously $\overline{Q} \subseteq \overline{P_t}$ for all $t \geq t_0$.
If $\bm{x}_{t_0}(r) = \bm{x}_{t_0}(r') \in Q$ for some $r, r' \in U$,
then a contradiction to the maximality of $A_{t_0}$ is derived;
activating $A_{t_0} \cup \{ r \}$ does not yield a goal configuration.
Thus, exactly one robot in $U$ resides at each $\bm{q} \in \overline{Q}$,
and hence $|\overline{Q}| = |U| > f$.

Let $r$ be any robot in $U$,
and let $\bm{q}$ be its destination at time $t_0$.
If $\bm{q} \not\in \overline{Q}$,
then a contradiction to the maximality of $A_{t_0}$ is derived;
activating $A_{t_0} \cup \{ r \}$ does not yield a goal configuration,
since $|\overline{(Q \setminus \{ \bm{x}_{t_0}(r) \})} \cup \{ \bm{q} \}| 
\geq (|Q| - 1) +1 = |U| > f$.
Thus, $\bm{q} \in \overline{Q}$.

Consider a directed graph $G = (\overline{Q}, E)$,
where $(\bm{p},\bm{q}) \in E$,
if and only if the destination of a robot $r \in U$ at $\bm{p}$ at time $t_0$ is $\bm{q}$.
Recall that there is a bijection between $U$ and $\overline{Q}$.
Since the outdegree of every vertex $\bm{p} \in \overline{Q}$ is 1,
there is a directed loop $L \subseteq \overline{Q}$ (which may be a self-loop) in $G$.
It is a contradiction to the maximality of $A_{t_0}$;
activating $A_{t_0} \cup L$ does not yield a goal configuration.
Thus, $|U| \leq f$.

Suppose that at $t_0$ all robots in $U$ crash,
and consider a schedule ${\cal S}'$ that activates 
$A_t$ for all $0 \leq t \leq t_0-1$,
and $A_t \cup U$ for all $t \geq t_0$.
Then the execution $\mathcal{E}'$ starting from $P_0$ under ${\cal S}'$
is exactly the same as $\mathcal{E}$,
and does not reach a goal configuration, despite that ${\cal S}'$ is fair; 
$\Phi$ is not an algorithm for the $f$FGP.
It is a contradiction.
\end{proof}

\section{Conclusions}
\label{Sconclusion}

We have proposed the minimum algorithm size (MAS) as an essential measure to
measure the complexity of a problem,
and showed an infinite hierarchy of the problems with respect to their MASs 
to justify this proposal;
the set of problems with MAS being $c$ is not empty for each integer $c > 0$ and $\infty$,
and hence target function is a resource irreplaceable, e.g., with time.
Then, under the $\cal SSYNC$ scheduler,
we have established the MASs for the self-stabilizing gathering and related problems.

Main results are summarized in Table~\ref{T0010}.
For example, we have clarified the difference between two similar
fault tolerant gathering problems, the FG and the FGP;
the MAS for the FG is 3, while the one for the FGP is $\infty$.
It might be interesting to observe that, for $2 \leq c \leq n$, 
the MAS is $c$ for the $c$SCT and the one for the $c$GAT is 1.
This fact is counter-intuitive,
since gathering objects (which decreases entropy) seems to be a harder task than
scattering objects (which increases entropy).

An open question asks how the time complexity decreases,
as the number of target functions that an algorithm can use increases.
It may be interesting if a kind of time acceleration theorem holds.

Another question asks the relation between the MAS of a problem $\Pi$
and $\sigma = \min_{G \in {\cal G}} \sigma(G)$,
where $\cal G$ is the set of $\Pi$'s goal configurations
and $\sigma(G)$ is the symmetricity of $G$.
Observe that identifiers are mainly used to break symmetry,
i.e., to decrease the symmetricity.

Finally, we conclude the paper by giving a list of some open problems.

\begin{enumerate}
\item
Complete Table~\ref{T0010} for the problems under the $\cal FSYNC$ scheduler.

\item 
What is the MAS for the gathering problem under the $\cal ASYNC$ scheduler?

\item
For a fixed pattern $G$,
what is the MAS for the pattern formation problem for $G$?

\item
What is the MAS for the $f$-fault tolerant 
convergence problem to $f$ points, for $f \geq 3$?

\item
What is the MAS for the Byzantine fault tolerant gathering problem?

\item
Characterize the problem whose MAS is $2$.

\item
Characterize the problem whose MAS is $\infty$.

\item
For a homonymous distributed network with topology $G$,
consider the number of identifiers necessary and sufficient to 
solve the leader election problem.

\end{enumerate}


\begin{thebibliography}{99}

\bibitem{AP04} %
N. Agmon and D. Peleg,
{\em Fault-tolerant gathering algorithms for autonomous mobile robots},
In Proc. 15th Annual ACM-SIAM Symposium on Discrete Algorithms,
pp.~1063--1071, 2004.


\bibitem{ADDDL20}
K. Altisen, A. K. Datta, S. Devismes, A. Durand, and L. L. Larmore,
{\em Election in unidirectional rings with homonyms,}
Journal of Parallel and Distributed Computing,
146, 79--95, 2010.


\bibitem{AOSY99} %
H. Ando, Y. Oasa, I. Suzuki, and M. Yamashita,
{\em A distributed memoryless point convergence algorithm for mobile robots
with limited visibility},
IEEE Trans. Robotics and Automation, 15, 818--828, 1999.

\bibitem{Angluin80} %
D. Angluin,
{\em Local and global properties in networks of processors,}
In Proc. 12th ACM Symposium on Theory of Computing,
pp.~82--93, 1980.

\bibitem{AADFP06} %
D. Angluin, J. Aspens, Z. Diamadi, M. J. Fisher, and R. Peralta,
{\em Computation in networks of passively mobile finite-state sensors,}
Distributed Computing, 18(4), 235--253, 2006.

\bibitem{AAIJR15}
S. Ar\'{e}valo, A. F. Anta, D. Imbs, E. Jim\'{e}nez, and M. Raynal,
{\em Failure detectors in homonymous distributed systems
(with an application to consensus),}
Journal of Parallel and Distributed Computing,
83, 83--95, 2015.

\bibitem{ASY22} %
Y. Asahiro, I. Suzuki, and M. Yamashita,
{\em Monotonic self-stabilization and its application 
to robust and adaptive pattern formation},
Theoretical Computer Science, 934, 21--46, 2022.

\bibitem{AY23} %
Y. Asahiro, and M. Yamashita,
{\em Compatibility of convergence algorithms for autonomous mobile robots,}
arXiv: 2301.10949, 2023. 
(See also the extended abstract appeared in 
Proc. 30th Colloquium on Structural Information and Communication Complexity,
Lecture Notes in Computer Science, 13892, pp.~149--164, 2023.)

\bibitem{AS91} %
H. Attiya, and M. Snir,
{\em Better computing on the anonymous ring,}
J. Algorithms, 12, 204--238, 1991.

\bibitem{ASW88} %
H. Attiya, M. Snir, and M. K. Warmuth,
{\em Computing on the anonymous ring,}
J. ACM, 35(4), 845--875, 1988.

\bibitem{BFFS02} %
L. Barri\'{e}re, P. Flocchini, P. Fraignaud, and N. Santoro,
{\em Capture of an intruder by mobile agents,}
In Proc. 14th Symposium on Parallel Algorithms and Architectures,
pp.~200--209, 2002.

\bibitem{BV97} %
P. Boldi, and S. Vigna,
{\em Computing vector functions on anonymous networks,}
In Proc. 4th Colloquium on Structural Information and Communication Complexity,
pp.~201--214, 1997.

\bibitem{BV01} %
P. Boldi, and S. Vigna,
{\em An effective characterization of computability in anonymous networks,}
In Proc. Int'l Symposium on Distributed Computing, pp. 33--47, 2001.


\bibitem{BDT13} %
Z. Bouzid, S. Das, and S. Tixeuil,
{\em Gathering of mobile robots tolerating multiple crash faults},
In Proc. IEEE 33rd Int'l Conference on Distributed Computing Systems,
pp.~337--346, 2013.

\bibitem{CSN19} %
S. Cicerone, G. Di Stefano, and A. Navarra,
{\em Asynchronous robots on graphs: Gathering,}
In \textit{Distributed Computing by Mobile Entities},
Lecture Notes in Computer Science, 11340, pp.~184--217, 2019.


\bibitem{CFPS12} %
M. Cieliebak, P. Flocchini, G. Prencipe, and N. Santoro,
{\em Distributed computing by mobile robots: gathering},
SIAM J. Computing, 41, 829--879, 2012.

\bibitem{CP05} %
R. Cohen and D. Peleg,
{\em Convergence properties of the gravitational algorithm in  asynchronous mobile robots,}
SIAM J. Computing, 34, 1516--1528, 2005.

\bibitem{CDFH11} %
A. Cord-Landwehr, B. Degener, M. Fischer, M H\"{u}llman, B. Kempkes,
A. Klaas, P. Kling, S. Kurras, M. M\"{a}rtens, F. Meyer auf der Heide,
C. Raupach, K. Swierkot, D. Warner, C. Weddemann, and D. Wonisch,
{\em A new approach for analyzing convergence algorithms for mobile robots},
%
In Proc. 38th Int'l Colloquium on Automata, Languages, and Programming, 
Part II, Lecture Notes in Computer Science, 6756, 
pp.~650--661, 2011

\bibitem{DHRS19} %
J. J. Daymude, K. Hinnenthal, A. W. Richa, and C. Scheideler,
{\em Computing by programmable particles,}
In \textit{Distributed Computing by Mobile Entities},
Lecture Notes in Computer Science, 11340, pp.~615--681, 2019.

\bibitem{DFSY15} %
S. Das, P. Flocchini, N. Santoro, and M Yamashita,
{\em Forming sequences of geometric patterns with oblivious mobile robots},
Distributed Computing, 28, 131--145, 2015.

\bibitem{DDGRS14} %
Z. Derakhshandeh, S. Dolve, R. Gmyr, A. W. Richa, C. Scheideler, T. Strothmann,
{\em Brief announcement: amoebot - a new model for programmable matter,}
In Proc. 26th ACM Symposium on Parallel Algorithms and Architectures,
pp.~220--222, 2014.

\bibitem{DPT19} %
X. D\'{e}fago, M. G. Potop-Butucaru, and S. Tixeuil,
{\em Fault-tolerant mobile robots},
In \textit{Distributed Computing by Mobile Entities},
Lecture Notes in Computer Science, 11340, pp.~234-251, 2019.

\bibitem{DFGKRT13}
C. Delporte-Gallet, H. Fauconnier, R. Guerraoui, A. Kermarrec,
E. Ruppert, and H. Tran-The, 
{\em Byzantine agreement with homonyms,}
Distributed Computing, 26, 321--340, 2013.


\bibitem{DFT14}
C. Delporte-Gallet, H. Fauconnier, and H. Tran-The,
{\em Leader election in rings with homonyms,}
Proc. Int'l Conf. Networked Systems, 9--24, 2014.




\bibitem{DP09} %
Y. Dieudonn\'{e} and F. Petit,
{\em Scatter of weak mobile robots,}
Parallel Processing Letters, 19(1), 175--184, 2009.


\bibitem{DP12} %
Y. Dieudonn\'{e} and F. Petit,
{\em Self-stabilizing gathering with strong multiplicity detection,}
Theoretical Computer Science, 428, 47--57, 2012.

\bibitem{Dolve00}
S. Dolve,
{\em Self-stabilization},
MIT Press, 2000.


\bibitem{DFPS01} %
S. Dobrev, P. Flonicchi, G, Prencipe, and N. Santoro,
{\em Mobile search for a black hole in an anonymous ring,}
In Proc. 15th Int'l Conference on Distributed Computing, pp.~166--179, 2001.


\bibitem{DP04}
S. Dobrev and A. Pelc,
{\em Leader election in rings with nonunique labels,}
Fundamenta Informaticae, 59(4), 333--347, 2004.


%
%
%
%
%

\bibitem{Flocchini19} %
P. Flocchini, {\em Gathering},
In \textit{Distributed Computing by Mobile Entities},
Lecture Notes in Computer Science, 11340, pp.~63--82, 2019.

\bibitem{FPS12} %
P. Flocchini, G. Prencipe, and N. Santoro, 
{\em Distributed Computing by Oblivious Mobile Robots,}
Synthesis Lectures on Distributed Computing Theory \#10,
Morgan \& Claypool Publishers, 2012.

\bibitem{FPSW99} %
P. Flocchini, G. Prencipe, N. Santoro, and P. Widmayer,
{\em Hard tasks for weak robots: The role of common knowledge in pattern
formation by autonomous mobile robots,}
In Proc. 10th Annual International Symposium on Algorithms and Computation,
Lecture Notes in Computer Science, 1741, 
pp.~93--102, 1999.

\bibitem{Ilcinkas19} %
D. Ilcinkas,
{\em Oblivious robots on graphs: Exploration,}
In \textit{Distributed Computing by Mobile Entities},
Lecture Notes in Computer Science, 11340, pp.~218--233, 2019.

\bibitem{ISKI12} %
T. Izumi, S. Soussi, Y. Katayama, N. Inuzuka, X. Defago, K. Wada, and M. Yamashita,
{\em The gathering problem for two oblivious robots with unreliable compasses},
SIAM J. Computing, 41, pp. 26--46, 2012.



\bibitem{Katreniak11} %
B. Katreniak,
{\em Convergence with limited visibility by asynchronous mobile robots},
In Proc. 18th International Colloquium on Structural Information and Communication Complexity (SIROCCO2011), 
Lecture Notes in Computer Science, 6786, pp.~125--137, 2011.

\bibitem{KKB90} %
E. Kranakis, D. Krizanc, and J. van den Berg,
{\em Computing boolean functions on anonymous networks,}
In Proc. 17th Int'l Colloquium on Automata, Languages, and Programming, 
Lecture Notes in Computer Science, 443,
pp.~254--267, 1990.

\bibitem{LKAV21} %
M. Lgothetis, G. Karras, L. Alevizos, C. Verginis, P. Roque,
K. Roditakis, A. Mkris, S. Garcia, P. Schillinger, A. Di Fava, P. Pelliccione,
A. Argyros, K. Kyriakopoulos, and D. V. Dimarogonas, 
{\em Efficient cooperation of heterogeneous robotic agents:
A decentralized framework,}
IEEE, Trans. Robotics and Automation Magazine, 28(2), 74--87, 2021.

\bibitem{LYKY18} %
Z. Liu, Y. Yamauchi, S. Kijima, and M. Yamashita,
{\em Team assembling problem for asynchronous heterogeneous mobile robots,}
Theoretical Computer Science, 721, pp. 27--41, 2018.

\bibitem{Luna19} %
G. A. Di Luna,
{\em Mobile agents on dynamic graphs,}
In \textit{Distributed Computing by Mobile Entities},
Lecture Notes in Computer Science, 11340, pp.~549--586, 2019.

\bibitem{Lynch96} %
N. Lynch,
{\em Distributed Algorithms,}
Morgan Kaufmann, 1996.

\bibitem{MA89} %
Y. Matias and Y. Afek,
{\em Simple and efficient election algorithms for anonymous networks,}
In Proc. 3rd Int'l Workshop on Distributed Algorithms,
pp.~183--194, 1989.

\bibitem{Prencipe19} %
G. Prencipe,
{\em Pattern formation},
In \textit{Distributed Computing by Mobile Entities},
Lecture Notes in Computer Science, 11340, pp.~37--62, 2019.

\bibitem{SY96} %
I. Suzuki and M. Yamashita,
{\em Distributed anonymous mobile robots -- formation and agreement problems,}
In Proc. 3rd Colloquium on Structural Information and Communication Complexity,
pp.~313--330, 1996.

\bibitem{SY99} %
I. Suzuki and M. Yamashita,
{\em Distributed anonymous mobile robots -- formation and agreement problems},
SIAM J. Computing, 
28, 1347--1363, 1999.

\bibitem{Viglietta19} %
G. Viglietta,
{\em Uniform circle formation},
In \textit{Distributed Computing by Mobile Entities},
Lecture Notes in Computer Science, 11340, pp.~83--108, 2019.

\bibitem{YK88} %
M. Yamashita and T. Kameda,
{\em Computing on an anonymous network,}
In Proc. 7th ACM Symposium on Principles of Distributed Computing,
pp.~117--130, 1988.

\bibitem{YK89}
M. Yamashita and T. Kameda,
{\em Electing a leader when processor identity numbers are not distinct
(extended abstract),}
Proc. Int'l Workshop WDAG'89, 303--314, 1989.

\bibitem{YK96} %
M. Yamashita and T. Kameda,
{\em Computing on anonymous networks I. 
Characterizing solvable cases,}
IEEE Trans. Parallel and Distributed Systems, 7(1), 69--89, 1996.

\bibitem{YK96b} %
M. Yamashita and T. Kameda,
{\em Computing functions on asynchronous anonymous networks,}
Mathematical Systems Theory, 29, 331--356, 1996.

\bibitem{YK99} %
M. Yamashita and T. Kameda,
{\em Leader election problem on networks in which processor identity numbers
are not distinct,}
IEEE Trans. Parallel and Distributed Systems, 10(9), 878--887, 1999.

\bibitem{YS10} %
M. Yamashita and I. Suzuki,
{\em Characterizing geometric patterns formable by oblivious anonymous mobile robots},
Theoretical Computer Science, 411, 2433--2453, 2010.

\bibitem{YUKY17} %
Y. Yamauchi, T. Uehara, S. Kijima, and M. Yamashita,
{\em Plane formation by synchronous mobile robots in the three-dimensional Euclidean space},
J. ACM, 64, 1--43, 2017.
\end{thebibliography}
\end{document}